\documentclass[draftcls,onecolumn]{IEEEtran}

\usepackage[utf8]{inputenc} 
\usepackage[T1]{fontenc}
\usepackage{url}
\usepackage{ifthen}
\usepackage{cite}
\usepackage[cmex10]{amsmath} 
\usepackage{graphicx}
\usepackage{amsmath, amsthm, amssymb, amsfonts}
\usepackage{tabularx}
\usepackage[linesnumbered,ruled]{algorithm2e}
\usepackage{algpseudocode}
\usepackage{subfigure}
\usepackage{url}
\usepackage{xcolor} 
\usepackage{multirow}
\usepackage{colortbl}

\newtheorem{definition}{Definition}
\newtheorem{lemma}{Lemma}
\newtheorem{corollary}{Corollary}
\newtheorem{theorem}{Theorem}

\newtheorem{lp}{Linear Programming}

\newcommand \sma {{\text{\tt sma}}}
\newcommand \row {{\sf row}}
\newcommand \col {{\sf col}}
\newcommand \rank {{\sf rank}}
\newcommand \Tr {{\sf T}}
\newcommand \wt {{\sf wt}}
\newcommand \one {{\mathbf{1}\!\!\!\mathbf{1}}}

\begin{document}

\title{Update Bandwidth for Distributed Storage}

\author{Zhengrui Li,
        Sian-Jheng Lin,~\IEEEmembership{Member,~IEEE,}
        Po-Ning Chen,~\IEEEmembership{Senior Member,~IEEE,}
        Yunghsiang S. Han,~\IEEEmembership{Fellow,~IEEE,}
        and Hanxu Hou,~\IEEEmembership{Member,~IEEE}
\thanks{This work was partially presented at the 2019 IEEE International Symposium on Information Theory.

This work was partially funded by the CAS Hundred Talents Program,  the NSFC of China (No.~61671007, 61701115), and Start Fund of Dongguan University of Technology (KCYXM2017025).

Z. Li and S.-J. Lin are with the School of Information Science and Technology, University of Science and Technology of China (USTC), Hefei, Anhui, China, email: dd622089@mail.ustc.edu.cn, sjlin@ustc.edu.cn.}  \thanks{P.-N. Chen is with the Department of Electrical and Computer Engineering, National Chiao Tung University, email:poning@faculty.nctu.edu.tw}          
\thanks{Y. S. Han and H. Hou are with the
School of Electrical Engineering \& Intelligentization,
Dongguan University of Technology, email:  yunghsiangh@gmail.com, houhanxu@163.com.
}

}

%
%

\markboth{Journal of \LaTeX\ Class Files,~Vol.~14, No.~8, August~2015}%
{Shell \MakeLowercase{\textit{et al.}}: Bare Demo of IEEEtran.cls for IEEE Journals}
%



\newcommand \gpri {{\tt g}}

\maketitle

\begin{abstract}
In this paper, we consider the update bandwidth in distributed storage systems~(DSSs). The update bandwidth, which measures the transmission efficiency of the update process in DSSs, is defined as the total amount of data symbols transferred in the network when the data symbols stored in a node are updated. 
This paper contains the following contributions. First, we establish the closed-form expression of the minimum update bandwidth attainable by irregular array codes.
Second, after defining a class of irregular array codes, called Minimum Update Bandwidth~(MUB) codes, which achieve the minimum update bandwidth of irregular array codes, we determine 
the smallest code redundancy attainable by MUB codes.
Third, the code parameters, with which the minimum code redundancy of irregular array codes and the smallest code redundancy of MUB codes can be equal, are identified, which allows us to define MR-MUB codes as a class of irregular array codes that simultaneously achieve the minimum code redundancy and the minimum update bandwidth.
Fourth, 
we introduce explicit code constructions of MR-MUB codes and MUB codes with the smallest code redundancy. 
Fifth, we establish a lower bound of the update complexity of MR-MUB codes,
which can be used to prove that 
the minimum update complexity of irregular array codes may not be achieved
by MR-MUB codes.
Last, we construct a class of $(n = k + 2, k)$ vertical maximum-distance separable (MDS) array codes that can achieve all of the minimum code redundancy, the minimum update bandwidth and the optimal repair bandwidth of irregular array codes.

\end{abstract}
\IEEEpeerreviewmaketitle

\section{Introduction}\label{sec:intro}

Some distributed storage systems (DSSs) adopt replication policy to improve reliability. However, the replication policy requires a high level of storage overhead. 
To reduce this overhead while maintain reliability, the erasure coding has been used in DSSs, such as Google File System~\cite{ford2010} and Microsoft Azure Storage~\cite{Huang2012}. 
A main issue of erasure codes in DSSs is the required bandwidth to repair failure node(s). To tackle this issue, many linear block codes, such as regenerating codes~\cite{dimikis2010,rashmi2011} and locally repairable codes~(LRCs)~\cite{gopalan2012,dimakis2014}, were proposed in recent years. When the original data symbols change, the coded symbols stored in a DSS must be updated accordingly. Since performing updates consumes both bandwidth and energy, a higher level of update efficiency is favorable for erasure codes in scenarios where updates are frequent.

The update process in a DSS has two important phases, which are symbol transmission among nodes and the symbol updating (i.e., reading-out and writing-in) in each node. Thus, the update efficiency should include the transmission efficiency and the I/O efficiency. Lots of update-efficient codes \cite{blaum1995,blaum1996,blaum1999,xu1999:b,jin09,mazumdar2012,antha2010,han2015} have been proposed to minimize the update complexity 
from the viewpoint of the I/O efficiency in the update process.
These works basically define the update complexity as the average number of coded symbols (i.e., parity symbols) that must be updated when any single data symbol is changed. 
Clearly,  when we consider updating one data symbol,  as the number of symbols transmitted between nodes is at most one, the less nodes affected by an update of a symbol, the better the transmission efficiency. Thus, the transmission efficiency problem seems simple and less important. However, when we consider updating many or even all symbols in a node simultaneously, the transmission efficiency problem becomes complicated and significant, and the above definition of the update complexity is not well suitable in this case. To our best knowledge, there is no work discussing the transmission efficiency in the update process of a DSS.  

In this paper, we introduce a new metric, called the {\it update bandwidth}, to measure the transmission efficiency in the update process of erasure codes applied in DSSs. It is defined as the average amount of symbols that must be transmitted among nodes when the data symbols stored in a node are updated.
As the storage capacity of a node is very large nowadays, we need to divide the data into small blocks of data symbols to encode. Since each block
is often self-contained in its structure, 
it is justifiably more efficient to have an updating operation
to operate on a block as a whole.
In other words, when any data symbol in a block is required to be updated, 
all data symbols in the block are involved in this single updating operation. As the update bandwidth is the main focus in this paper,
without loss of generality,
we consider the simplest setting that there is only one coded block in each node in our analysis.

\begin{figure}
\center
   \includegraphics[width=0.7\columnwidth]{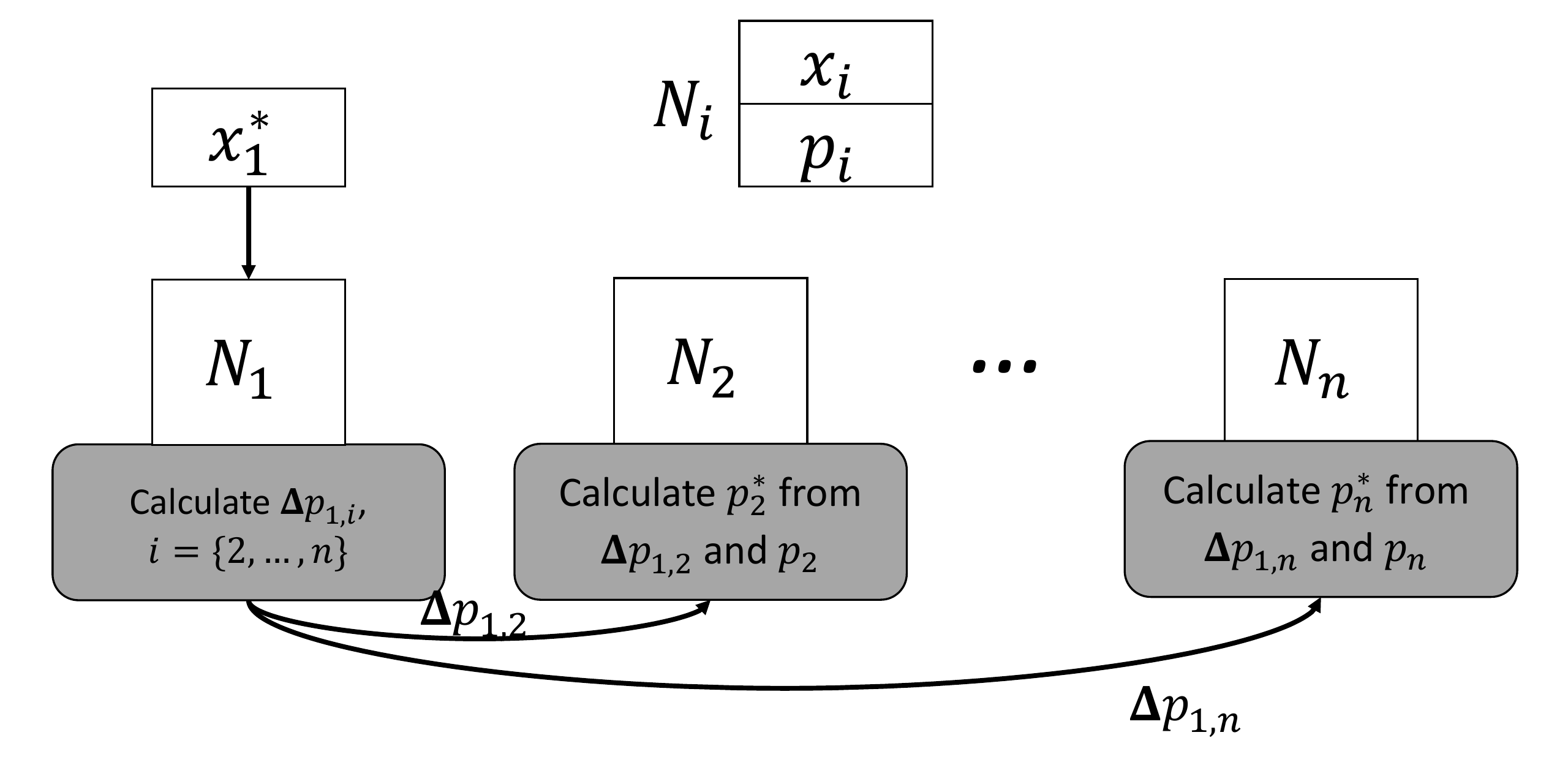}
\caption{\label{fig:fig1} An instance of the considered update model,
where the data symbols stored in $N_{1}$ are updated, and $N_{i}$ has the data vector $\mathbf{x}_{i}$ and the parity vector  $\mathbf{p}_{i}$. When updating the data symbols stored in $N_{1}$, $N_{1}$ sends intermediate symbols $\{\Delta \mathbf{p}_{1,i}\}_{i=2,\dots,n}$ respectively to all other nodes such that they can calculate the new parity vectors.}
\end{figure}

The update model that we consider is described as follows. Assume that there are $n$ nodes $\{N_i\}_{i=1}^{n}$ in the network. 
Node $N_{i}$ stores data vector $\mathbf{x}_{i}$ and parity vector $\mathbf{p}_{i}$, where the former consists of data symbols, while the parity symbols are placed in the latter. Fig.~\ref{fig:fig1} demonstrates the update procedure when the data vector $\mathbf{x}_{1}$ is updated to $\mathbf{x}_{1}^{*}$. 
In the update procedure, $N_{1}$ first calculates $n-1$ intermediate vectors $\{\Delta\mathbf{p}_{1,i}\}_{i=2}^n$, and then send $\Delta\mathbf{p}_{1,i}$ to $N_{i}$ respectively for $i=2,3,\ldots , n$. 
After receiving $\Delta\mathbf{p}_{1,i}$, $N_{i}$ computes the updated parity vector $\mathbf{p}_{i}^{*}$ from $\Delta\mathbf{p}_{1,i}$ and the old parity vector $\mathbf{p}_{i}$. This completes the update procedure. 
Notably, this update model is similar to the one adopted in \cite{zhang2012}, in which partial-updating schemes for erasure-coded storage are  considered. Our general update model will be given formally in Section \ref{sec:ub}. 

\begin{figure} 
  \centering
    \subfigure[]{\label{fig:side:a}
    \label{fig:subfig:a} 
    \begin{tabular}{|c|c|c|c|}
\hline
\rowcolor{gray}$x_{1,1}$ & $x_{2,1}$ & $x_{3,1}$ & $x_{4,1}$\\
\hline
\rowcolor{gray}$x_{1,2}$ & $x_{2,2}$ & $x_{3,2}$ & $x_{4,2}$\\
\hline
$x_{3,1}+x_{4,1}$ & $x_{1,1}+x_{3,1}$ & $x_{2,1}+x_{4,1}$ & $x_{1,1}+x_{2,1}$\\
\hline
$x_{3,2}+x_{4,2}$ & $x_{1,2}+x_{3,2}$ & $x_{2,2}+x_{4,2}$ & $x_{1,2}+x_{2,2}$\\
\hline
\end{tabular}}
  \subfigure[]{\label{fig:side:b}
    \label{fig:subfig:b} 
    \begin{tabular}{|c|c|c|c|}
\hline
\rowcolor{gray}$x_{1,1}$ & $\underline{x_{2,1}}$ & $x_{3,1}$ & $\underline{x_{4,1}}$\\
\hline
\rowcolor{gray}$x_{1,2}$ & $\underline{x_{2,2}}$ & $\underline{x_{3,2}}$ & $x_{4,2}$\\
\hline
$x_{4,1}+x_{3,2}$ & $x_{1,1}+x_{4,2}$ & $\underline{x_{2,1}+x_{1,2}}$ & $x_{3,1}+x_{2,2}$\\
\hline
$x_{2,1}+x_{2,2}+x_{3,2}$ & $x_{3,1}+x_{3,2}+x_{4,2}$ & $x_{4,1}+x_{4,2}+x_{1,2}$ & $\underline{x_{1,1}+x_{1,2}+x_{2,2}}$\\
\hline
\end{tabular}
}
\caption{(a) presents an instance, where gray rows contain the data symbols, of a $4 \times 4$ P-code with optimal update complexity $2$ and update bandwidth $4$; (b) presents an instance, where gray rows contain the data symbols, of a proposed  $(n=4,k=2)$ codes, which has update complexity larger than $2$ and minimum update bandwidth $3$. Note that the instance presents in (b) also has the optimal repair bandwidth.}
\label{fig:fig2} 
\end{figure} 

It is worthy mentioning that the codes with the minimum update complexity (i.e., I/O efficiency) may not achieve the minimum update bandwidth, and vice versa. 
To show that, Fig. \ref{fig:fig2} presents two $(n=4,k=2)$ maximum distance separable (MDS) array codes, where the elements in the $i$-th column are the symbols stored in node $N_{i}$ and the number of symbols in each node is $\alpha=4$. 
In Fig. \ref{fig:side:a}, the first row and the third row form an instance of a $2 \times 4$ P-code~\cite{jin09}, and the second row and the fourth row form another instance of a $2 \times 4$ P-code. Thus, Fig. \ref{fig:side:a} is an instance of a $4 \times 4$ P-code. Furthermore, Fig. \ref{fig:side:b} is an instance of our codes proposed in Section \ref{sec:code}.
In Figs. \ref{fig:side:a} and \ref{fig:side:b}, the data symbols $\{x_{i,j}\}_{i=1,\dots,4,j=1,2}$ are arranged in the first two {gray} rows, and the last two rows are occupied by parity symbols. It can be verified that the data symbols can be recovered by accessing any two columns of the codes in Figs. \ref{fig:side:a} and \ref{fig:side:b}, and hence $k=2$.
It is known that P-codes \cite{jin09} achieve the minimum update complexity when $n-k=2$. 
Hence, when updating a data symbol in Fig.~\ref{fig:side:a}, we must update at least $n-k=2$ parity symbols. For example, when updating $x_{1,1}$, the third symbol in the second column $x_{1,1}+x_{3,1}$ and the third symbol in the fourth column $x_{1,1}+x_{2,1}$ need to be updated. 
However, in Fig. \ref{fig:side:b}, when updating a data symbol, two or three parity symbols need to be updated, i.e., the corresponding update complexity is larger than $2$. For example, when updating $x_{1,1}$, the third symbol in the second column $x_{1,1}+x_{4,2}$ and the fourth symbol in the fourth column $x_{1,1}+x_{1,2}+x_{2,2}$ need to be updated. Yet, the updating of $x_{1,2}$
requires the modification of both parity symbols in the third column (i.e, $x_{2,1}+x_{1,2}$ and $x_{4,1}+x_{4,2}+x_{1,2}$) and the fourth symbol in the fourth column (i.e., $x_{1,1}+x_{1,2}+x_{2,2}$). 

Next, we consider the update bandwidth. Suppose that the two data symbols in the first node in Fig.~\ref{fig:side:a} are updated, i.e., $x_{1,j}$ are updated to  $x_{1,j}^{*}$, $j=1,2$. 
The first node  should send two symbols $\Delta x_{1,1}$ and $\Delta x_{1,2}$ to both nodes $2$ and $4$, where  $\Delta x_{i,j}=x_{i,j}^{*}-x_{i,j}$. 
Thus, the required bandwidth is four. 
It is easy to check that the required bandwidth of updating two data symbols of any other node is also four. 
Therefore, the update bandwidth of the $4 \times 4$ P-code is four. Next we 
show that the update bandwidth of the code in Fig. \ref{fig:side:b} is three. When two data symbols in node $1$ in Fig.~\ref{fig:side:b} are updated, we only need to send $\Delta x_{1,1}$ to node 2, $\Delta x_{1,2}$ to node $3$, and $(\Delta x_{1,1} + \Delta x_{1,2})$ to node $4$. Therefore, the required update bandwidth is three. 
We can verify that the required update bandwidth when updating any other node 
in Fig.~\ref{fig:side:b} is also three. Consequently, the update bandwidth 
of the code in Fig.~\ref{fig:side:b} is better than that of the $4\times 4$ P-code
in Fig.~\ref{fig:side:a}.
We will show in Section \ref{sec:MRMB} that the code in Fig. \ref{fig:side:b} achieves the minimum update bandwidth among all $(4,2)$ irregular array codes 
with two data symbols per node.

Other than update complexity and update bandwidth, the repair bandwidth, defined as the amount of symbols downloaded from the surviving nodes to repair the failed node, is also an important consideration in DSSs.
The repair problem was first brought into the spotlight by Dimakis  et al.~\cite{dimikis2010}. 
It can be anticipated that a well-designed code with both minimum update bandwidth and optimal repair bandwidth is attractive for DSSs. Surprisingly, the code in Fig. \ref{fig:side:b} also achieves the optimal repair bandwidth among all $(4,2)$ MDS array codes.
One can check that we can repair the four symbols stored in node $1$ by downloading the six underlined symbols in Fig. \ref{fig:side:b}, i.e., $\{x_{2,1},x_{2,2}\}$ from node $2$, $\{x_{3,2}, x_{2,1}+x_{1,2}\}$ from node $3$, and $\{x_{4,1}, x_{1,1}+x_{1,2}+x_{2,2}\}$ from node $4$. Thus, the repair bandwidth of node $1$ is six, which is optimal for the parameters of $n=4$, $k=2$ and $\alpha=4$~\cite{dimikis2010}. We can verify that the repair bandwidth of any other node in Fig. \ref{fig:side:b} is also six. 
Therefore, the repair bandwidth of the code in Fig. \ref{fig:side:b} is optimal. 

The contributions of this paper are as follows. 
\begin{itemize}
\item We introduce a new metric, i.e., update bandwidth, and emphasize its importance in scenarios where storage updates are frequent.

\item We consider irregular array codes with a given level of protection against block erasures~\cite{tosato2014}, and establish the closed-form expression of the minimum update bandwidth attainable for such codes.
\item Referring the class of irregular array codes that achieve the minimum update bandwidth as MUB codes, we next derive
the smallest code redundancy attainable by MUB codes.
\item Comparing the smallest code redundancy of MUB codes
with the minimum code redundancy of irregular array codes
derived in \cite{tosato2014}, we identify a class of MUB codes, called MR-MUB codes, that can achieve simultaneously 
the minimum code redundancy of irregular array codes 
and the minimum update bandwidth of irregular array codes.


\item Systematic code constructions for MR-MUB codes and for 
MUB codes with the smallest code redundancy are both provided.

\item  We establish a lower bound of the update complexity of MR-MUB codes, 
by which we confirm that the update complexity of irregular array codes
may not be achieved by MR-MUB codes.

\item We construct an $(n,k=n-2)$ MR-MUB code with the optimal repair bandwidth for all nodes via the transformation in~\cite{transform}, 
confirming the existence of the irregular array codes that can simultaneously achieve 
the minimum code redundancy, the minimum update bandwidth
and the optimal repair bandwidth for all codes.

\end{itemize}

The rest of this paper is organized as follows. Section \ref{sec:def} introduces  the notations used in this paper and the proposed update model. Section \ref{sec:UBMD} establishes the necessary condition for the existence of an irregular array code. In Section \ref{sec:MRMB}, 
via the form of linear programmings,
we determine the minimum update bandwidth of irregular array codes
and the smallest code redundancy of MUB codes.
Section \ref{sec:code} presents the explicit constructions of MR-MUB codes and MUB codes. 
Section \ref{sec:UC} derives a lower bound of the update complexity of MR-MUB codes. Section \ref{sec:6} devises 
a class of $(n=k+2,k)$ MR-MUB codes with the optimal repair bandwidth for all nodes. Section \ref{sec:conclusion} concludes this work. 

\section{Preliminary}\label{sec:def}

\subsection{Definition}

We first introduce the notations used in this paper.
Let $[n]\triangleq\{1,\dots,n\}$ for a positive integer $n$. $(a_{i})_{i \in [n]}$ denotes an index set $(a_{1},a_{2},\dots, a_{n})$. Let $[x_{i,j}]_{i\in [m],j\in [n]}$ denote an $m \times n$ matrix whose entry in row $i$ and column $j$ is $x_{i,j}$. $\wt(\mathbf{v})$ denotes the weight of vector $\mathbf{v}$, i.e., the number of nonzero elements in vector $\mathbf{v}$. $\mathbf{M}^{\Tr}$ represents the transpose of matrix $\mathbf{M}$. $\row(\mathbf{M})$, $\col(\mathbf{M})$ and $\rank(\mathbf{M})$ represent the number of rows of $\mathbf{M}$, the number of columns of $\mathbf{M}$ and the rank of $\mathbf{M}$, respectively. 
$\mathbf{M}^{-1}$ denotes the inverse matrix of $\mathbf{M}$, provided $\mathbf{M}$ is invertible. $|S|$ denotes the cardinality of a set $S$.  $\mathbb{F}_{q}$ denotes the finite field of size $q$, where $q$ is a power of a prime. 
For two discrete random variables $X$ and $Y$, 
their joint probability distribution is denoted as $P_{XY}(x,y)$. $H_{q}(X)$ denotes the $q$-ary entropy of $X$, and $I_{q}(X; Y)$ denotes the $q$-ary mutual information  between $X$ and $Y$, where $q$ is the base of the logarithm. 
We consider linear codes throughout the paper and the main notations used in this paper are listed in Table~\ref{tab:1}.


\begin{table}
\begin{center}
\caption{\label{tab:1}Main notations used in this paper}
\begin{tabular}{l|l}
\hline\hline
Notation & Description\\
\hline
$n$ & The number of nodes\\
\hline
$m_{i}$ & The number of data symbols in node $i$\\
\hline
$p_{i}$ & The number of parity symbols in node $i$\\
\hline
$\mathbf{m}$ &  $\mathbf{m}=[m_{1} \dots m_{n}]^{\Tr}$\\
\hline
$\mathbf{p}$ &   $\mathbf{p}=[p_{1} \dots p_{n}]^{\Tr}$\\
\hline
$\mathbf{x}_{i}$ & The $i$-th data vector\\
\hline
$\mathbf{p}_{i}$ & The $i$-th parity vector\\
\hline
$\mathbf{c}_{i}$ & $\mathbf{c}_{i}=[\mathbf{x}_{i}^{\Tr} \quad \mathbf{p}_{i}^{\Tr}]^{\Tr}$, the $i$-th column vector\\
\hline
$\mathbf{C}$ & $\mathbf{C}=(\mathbf{c}_{1},\mathbf{c}_{2},\dots,\mathbf{c}_{n})$, the codeword of an irregular array code\\
\hline
$\mathbf{M}_{i,j}$ & The construction matrix\\
\hline 
$\mathbf{A}_{i,j}$, $\mathbf{B}_{i,j}$ & A full rank decomposition of $\mathbf{M}_{i,j}$, i.e., $\mathbf{M}_{i,j}=\mathbf{B}_{i,j}\mathbf{A}_{i,j}$\\
\hline
$\mathcal{E}$ & A subset of $[n]$ with $|\mathcal{E}|=n-k$,  where the elements in it are denoted as $e_{i}$ with $i \in [n-k]$ \\
\hline
$\bar{\mathcal{E}}$ & $\bar{\mathcal{E}}=[n] \setminus \mathcal{E}$, where the elements in it are denoted as $\bar{e}_{i}$ with $i \in [k]$\\
\hline
$\mathbf{C}_{\mathcal{E}}$, $\mathbf{X}_{\mathcal{E}}$, $\mathbf{P}_{\mathcal{E}}$ & $\mathbf{C}_{\mathcal{E}}=[\mathbf{c}_{e_{1}}^{\Tr} \dots \mathbf{c}_{e_{n-k}}^{\Tr}]^{\Tr}$,\ \ $\mathbf{X}_{\mathcal{E}}=[\mathbf{x}_{e_{1}}^{\Tr} \dots \mathbf{x}_{e_{n-k}}^{\Tr}]^{\Tr}$,\ \ $\mathbf{P}_{\mathcal{E}}=[\mathbf{p}_{e_{1}}^{\Tr} \dots \mathbf{p}_{e_{n-k}}^{\Tr}]^{\Tr}$\\
\hline 
$B$ & The number of data symbols\\
\hline 
$R$ & The number of parity symbols, i.e., code redundancy\\
\hline
$\gamma_{i,j}$ & The minimum number of symbols sent from node $i$ to node $j$ when updating the data symbols in node $i$\\
\hline 
$\gamma$ & The average required bandwidth when updating a node, i.e., update bandwidth\\
\hline
$\gamma_{\min}$ & The minimum update bandwidth among all irregular array codes\\
\hline
$R_{\min}$ & The minimum code redundancy among all irregular array codes\\
\hline
$R_{\sma}$ & The smallest code redundancy for irregular array codes with 
update bandwidth equal to $\gamma_{\min}$\\
\hline 
$\theta$ & The average number of parity symbols affected by a change of a single data symbol, i.e., update complexity\\
\hline\hline 
\end{tabular}
\end{center}
\end{table}


\subsection{Irregular array code}\label{sec:2a}
An  irregular array code \cite{tosato2014,Chen2019} can be represented as an irregular array. Formally, given a positive integer $n$ and two column vectors $\mathbf{m}=[m_{1} \ \dots \ m_{n}]^{\Tr}$ and $\mathbf{p}=[p_{1}\ \dots \ p_{n}]^{\Tr}$, where $m_{i}, p_{i} \geq 0$ for $i \in [n]$, the codeword of an irregular array code $\mathcal{C}$ over $\mathbb{F}_{q}$ is denoted as
\begin{equation}
\mathbf{C}=(\mathbf{c}_{1},\mathbf{c}_{2},\dots,\mathbf{c}_{n}),
\end{equation}
where the column vector $\mathbf{c}_{i}$ contains $m_{i}$ data symbols and $p_{i}$ parity symbols. Specifically, we denote
\begin{equation}\label{eq:ci}
\mathbf{c}_{i}=\begin{bmatrix}
\mathbf{x}_{i}\\
\mathbf{p}_{i}
\end{bmatrix}, \quad i \in [n],
\end{equation}
where $\mathbf{x}_{i}$ is the $i$-th data vector that contains $m_{i}$ data symbols and $\mathbf{p}_{i}$ is the $i$-th parity vector that contains $p_{i}$ parity symbols. Since $\mathbf{x}_{i}$ contains data symbols, we can naturally consider that $\mathbf{x}_{i}$ is uniformly distributed over $\mathbb{F}_{q}^{m_{i}}$, and $\mathbf{x}_{i}$ and $\mathbf{x}_{j}$ are independent for $i \neq j \in [n]$. As such, we have
\begin{equation}\label{eq:xi}
H_{q}(\mathbf{x}_i)=m_i \qquad \forall i \in [n],
\end{equation}
\begin{equation}\label{eq:ij}
I_{q}(\mathbf{x}_i; \mathbf{x}_j)=0 \qquad \forall i,j \in [n], i \neq j.
\end{equation}
 As all symbols in $\mathbf{p}_{i} \in \mathbb{F}_{q}^{p_{i}}$ may not be independent, we can only obtain
\begin{equation}\label{eq:q2}
H_{q}(\mathbf{p}_i) \leq p_i \qquad \forall i \in [n].
\end{equation}
The storage redundancy (i.e., code redundancy) of $\mathcal{C}$ is the total number of parity symbols, i.e., $R=\sum_{i=1}^{n}{p_i}$. 
An example of irregular array codes is illustrated in Fig. \ref{fig:fig3},
where the gray cells contain the data symbols. In this example, we have $\mathbf{m}=[4 \ 2 \ 2 \ 0]^{\Tr}$, $\mathbf{p}=[2 \ 3 \ 3 \ 3]^{\Tr}$ and $R=11$.
In addition, the first column of the irregular array code in Fig. \ref{fig:fig3} stores four data symbols $\mathbf{x}_{1}=[x_{1,1} \ x_{1,2} \ x_{1,3} \ x_{1,4}]^{\Tr}$ and two parity symbols $\mathbf{p}_1=[x_{2,1}+x_{2,2} \ x_{3,2}]^{\Tr}$. 

\begin{figure} 
  \centering
    \begin{tabular}{|c|c|c|c|}
\hline
\cellcolor{gray}{$x_{1,1}$} & \cellcolor{gray}$x_{2,1}$ & \cellcolor{gray}$x_{3,1}$ & $x_{1,1}+x_{1,3}+x_{3,1}$\\
\hline
\cellcolor{gray}$x_{1,2}$ & \cellcolor{gray}$x_{2,2}$ & \cellcolor{gray}$x_{3,2}$ & $x_{1,2}+x_{1,4}+x_{3,1}$\\
\hline
\cellcolor{gray}$x_{1,3}$ & $x_{1,1}$ & $x_{1,3}$ & $x_{2,2}+x_{3,1}$\\
\hline
\cellcolor{gray}$x_{1,4}$ & $x_{1,2}$ & $x_{1,4}$ \\
\cline{1-3}
$x_{2,1}+x_{2,2}$ & $x_{3,1}+x_{3,2}$ & $x_{2,1}$\\
\cline{1-3}
$x_{3,2}$ \\
\cline{1-1}
\end{tabular}
\caption{This figure, where the gray cells contain the data symbols, shows a $(4,2,\mathbf{m})$ irregular MDS array code with $\mathbf{m}=[4 \ 2 \ 2 \ 0]^{\Tr}$ and $\mathbf{p}=[2 \ 3 \ 3 \ 3]^{\Tr}$.}
\label{fig:fig3} 
\end{figure} 

When $m_{i}+p_{i}=m_{j}+p_{j}$ for any $i\neq j \in [n]$, the irregular array code $\mathcal{C}$ is reduced to a regular array code.  
When $m_{i}=m_{j}$ and $p_{i}=p_{j}$ for all $i\neq j \in [n]$, $\mathcal{C}$ is called a vertical array code. As an example, both codes in Fig. \ref{fig:fig2} are vertical array codes. When $p_{i}=0$ for $i \in [k]$ and $m_{j}=0$ for $k<j \leq n$, $\mathcal{C}$ is called a horizontal array code.
If we can retrieve all the data symbols by accessing any $k$ columns, and there is a set of $k-1$ columns which we can not retrieve all the data symbols from, then the code $\mathcal{C}$ is parameterized as an $(n,k,\mathbf{m})$ irregular  array code. 

We will demonstrate in Section~\ref{sec:6b} that the code in Fig. \ref{fig:fig3} can only be reconstructed by accessing at least any two columns, and hence 
it is a $(4,2,\mathbf{m})$ irregular array code. 
In fact, the code in Fig.~\ref{fig:fig3} is also an MUB code with the smallest code redundancy (cf.~Section~\ref{sec:4c}).
In the subsection that follows, we will introduce the update model and the update bandwidth
of $(n,k,\mathbf{m})$ irregular array codes.


\subsection{Update model and update bandwidth}\label{sec:ub}
In an $(n,k,\mathbf{m})$ irregular array code $\mathcal{C}$, each parity symbol can be generated as a linear combination of all  data symbols. Thus, the parity symbols in each column can be obtained from
\begin{equation}\label{eq:constr}
\mathbf{p}_{j}=\sum_{i=1}^{n}{\mathbf{M}_{i,j}\mathbf{x}_{i}} \quad j \in [n],
\end{equation}
where $\mathbf{M}_{i,j}$ is a $p_{j} \times m_{i}$ matrix, called construction matrix. Apparently, when the data vector in node $i$ is updated from $\mathbf{x}_{i}$ to $\mathbf{x}_{i}^{*}$, node $j$ with $j\in [n]\setminus \{i\}$ needs to update its parity vector via $\mathbf{p}_{j}^{*}=\mathbf{p}_{j}+\mathbf{M}_{i,j}\Delta\mathbf{x}_{i}$, where $\Delta\mathbf{x}_{i}=\mathbf{x}_{i}^{*}-\mathbf{x}_{i}$. 
Such update process can be divided into two steps. 
First, node $i$ calculates the intermediate vector $\mathbf{A}_{i,j}\Delta\mathbf{x}_{i}$, and sends these symbols to node $j$. 
Second, node $j$ calculates $\Delta\mathbf{p}_{j}=\mathbf{p}^{*}_{j}-\mathbf{p}_{j}$ from the intermediate vector via a linear transformation, i.e., $\Delta\mathbf{p}_{j}=\mathbf{B}_{i,j}\mathbf{A}_{i,j}\Delta\mathbf{x}_i$. 
As a result of \eqref{eq:constr}, the two matrices $\mathbf{A}_{i,j}$ and $\mathbf{B}_{i,j}$ corresponding to the linear transformations respectively performed by node $i$ and node $j$ must satisfy $\mathbf{M}_{i,j}=\mathbf{B}_{i,j}\mathbf{A}_{i,j}$. Based on the above update model, the number of symbols sent from node $i$ to node $j$ is $\row(\mathbf{A}_{i,j})$. 

Denoting 
\begin{equation}\label{eq:gamaij}
\gamma_{i,j} \triangleq \min \limits_{\mathbf{A}_{i,j},\mathbf{B}_{i,j}} \{\row(\mathbf{A}_{i,j})|\mathbf{M}_{i,j}=\mathbf{B}_{i,j}\mathbf{A}_{i,j}\}, \mbox{ for } i\neq j,
\end{equation}
as the minimum amount of symbols sent from node $i$ to node $j$ when updating the data symbols stored in node $i$, we have the following theorem.


\begin{theorem}\label{th:0}
$\row(\mathbf{A}_{i,j})=\gamma_{i,j}$ if, and only if,  
 $\rank(\mathbf{A}_{i,j})=\rank(\mathbf{B}_{i,j})=\row(\mathbf{A}_{i,j})$, and both $\mathbf{A}_{i,j}$ and $\mathbf{B}_{i,j}$ are with full rank. Furthermore,  $\gamma_{i,j}=\rank(\mathbf{M}_{i,j})$. 
\end{theorem}
\begin{proof}
We first prove that $\row(\mathbf{A}_{i,j})=\gamma_{i,j}$ implies 
 $\rank(\mathbf{A}_{i,j})=\rank(\mathbf{B}_{i,j})=\row(\mathbf{A}_{i,j})$, and both $\mathbf{A}_{i,j}$ and $\mathbf{B}_{i,j}$ have full rank, which in turns validates
 $\gamma_{i,j}=\rank(\mathbf{M}_{i,j})$.
 
Assuming that  $\row(\mathbf{A}_{i,j})$ is with the minimum value, i.e., $\row(\mathbf{A}_{i,j})=\gamma_{i,j}$, we show by contradiction that  $\row(\mathbf{A}_{i,j})=\rank(\mathbf{A}_{i,j})$. 
Suppose $\rank(\mathbf{A}_{i,j})<\row(\mathbf{A}_{i,j})$. Then, there is an invertible matrix $\mathbf{R}_{i,j}$ satisfying $\mathbf{R}_{i,j}\mathbf{A}_{i,j}=\tiny\begin{bmatrix}
\mathbf{A}_{i,j}'\\
[\mathbf{0}]
\end{bmatrix}$, where $[\mathbf{0}]$ is a $(\row(\mathbf{A}_{i,j})-\rank(\mathbf{A}_{i,j})) \times m_{i}$ zero matrix, and $\rank(\mathbf{A}_{i,j}')=\row(\mathbf{A}_{i,j}')=\rank(\mathbf{A}_{i,j})$. We thus have $\mathbf{M}_{i,j}=\mathbf{B}_{i,j}\mathbf{R}_{i,j}^{-1}\mathbf{R}_{i,j}\mathbf{A}_{i,j}=\mathbf{B}_{i,j}\mathbf{R}_{i,j}^{-1} \tiny\begin{bmatrix}
\mathbf{A}_{i,j}'\\
[\mathbf{0}]
\end{bmatrix}$, which implies $\mathbf{M}_{i,j}=\mathbf{B}_{i,j}'\mathbf{A}_{i,j}'$ with $\mathbf{B}_{i,j}'$ being the first $\rank(\mathbf{A}_{i,j})$ columns of $\mathbf{B}_{i,j}\mathbf{R}_{i,j}^{-1}$. However, $\rank(\mathbf{A}_{i,j}')=\rank(\mathbf{A}_{i,j})< \gamma_{i,j}$ contradicts to the definition of $\gamma_{i,j}$. We therefore confirm that if $\row(\mathbf{A}_{i,j})=\gamma_{i,j}$, then $\row(\mathbf{A}_{i,j})=\rank(\mathbf{A}_{i,j})$. 
Similarly, we can show that if $\col(\mathbf{B}_{i,j})=\gamma_{i,j}$, then
$\col(\mathbf{B}_{i,j})=\rank(\mathbf{B}_{i,j})$. As $\row(\mathbf{A}_{i,j})=\col(\mathbf{B}_{i,j})$, we conclude that
$\row(\mathbf{A}_{i,j})=\gamma_{i,j}$ implies
$\rank(\mathbf{A}_{i,j})=\rank(\mathbf{B}_{i,j})=\row(\mathbf{A}_{i,j})=\col(\mathbf{B}_{i,j})$, and both $\mathbf{A}_{i,j}$ and $\mathbf{B}_{i,j}$ have full rank.
An immediate consequence of the above proof is that this pair of $\mathbf{A}_{i,j}$ and $\mathbf{B}_{i,j}$ is a minimizer of \eqref{eq:gamaij}.
By Sylvester's rank inequality, we have 
\begin{equation}
\rank(\mathbf{B}_{i,j})+\rank(\mathbf{A}_{i,j})-\row(\mathbf{A}_{i,j})=\gamma_{i,j}\le\rank(\mathbf{M}_{i,j}).
\end{equation}
It can also be inferred that
\begin{equation}
\rank(\mathbf{M}_{i,j})\le \min\{\rank(\mathbf{B}_{i,j}),\rank(\mathbf{A}_{i,j})\}=\gamma_{i,j}.
\end{equation}
Hence, 
\begin{equation}
\gamma_{i,j}=\rank(\mathbf{M}_{i,j}).
\end{equation}

We next show the converse statement, i.e., if both $\mathbf{A}_{i,j}$ and $\mathbf{B}_{i,j}$ are with full rank and $\rank(\mathbf{A}_{i,j})=\rank(\mathbf{B}_{i,j})=\row(\mathbf{A}_{i,j})$, then $\row(\mathbf{A}_{i,j})=\gamma_{i,j}$. 
Given $\rank(\mathbf{B}_{i,j})=\row(\mathbf{A}_{i,j})$, we  obtain by Sylvester's rank inequality that
\begin{equation}
\rank(\mathbf{B}_{i,j})+\rank(\mathbf{A}_{i,j})-\row(\mathbf{A}_{i,j})=\rank(\mathbf{A}_{i,j})\le\rank(\mathbf{M}_{i,j})=\gamma_{i,j},
\end{equation}
which, together with $\gamma_{i,j}=\rank(\mathbf{M}_{i,j})\le\rank(\mathbf{A}_{i,j})$,
establishes $\row(\mathbf{A}_{i,j})=\rank(\mathbf{A}_{i,j})=\gamma_{i,j}$.
This completes the proof.
\end{proof}

Theorem~\ref{th:0} indicates that $\gamma_{i,j}=\rank(\mathbf{M}_{i,j})$ is the minimum amount of symbols required to be sent from node $i$ to node $j$ when updating the data symbols stored in node $i$. We thus define
the update bandwidth $\gamma$ for a code $\mathcal{C}$ as the average required bandwidth, i.e., 
\begin{equation}\label{eq:defgamma}
\gamma=\frac{1}{n}\sum_{i=1}^{n}{\sum_{j\in [n] \setminus \{i\}}{\gamma_{i,j}}}.
\end{equation}
By Theorem~\ref{th:0}, the update bandwidth $\gamma$ can be achieved by adopting two full-rank matrices 
that fulfill $\mathbf{M}_{i,j}=\mathbf{B}_{i,j}\mathbf{A}_{i,j}$,
where 
$\mathbf{B}_{i,j}$ is a $p_{j} \times \gamma_{i,j}$ matrix and $\mathbf{A}_{i,j}$ is a $\gamma_{i,j} \times m_{i}$ matrix. 
In the rest of the paper, the full-rank matrices $\mathbf{B}_{i,j}$ and $\mathbf{A}_{i,j}$ used in our update model are fixed as the ones with rank $\gamma_{i,j}$. 


\subsection{Encoding aspect of the update model}\label{sec:enc}

The update model in the previous subsection can also be 
equivalently characterized via an encoding aspect from \eqref{eq:gamaij}.
Specifically, we can first calculate
\begin{equation}\label{eq:gi}
\mathbf{p}_{i,j} = \mathbf{A}_{i,j}\mathbf{x}_{i}  \qquad \forall i,j \in [n], i\neq j.
\end{equation}
Similar to \eqref{eq:q2}, since symbols in $\mathbf{p}_{i,j} \in \mathbb{F}_{q}^{\gamma_{i,j}}$ are possibly dependent, 
we can only obtain
\begin{equation}\label{eq:q1}
H_{q}(\mathbf{p}_{i,j}) \leq\gamma_{i,j} \qquad \forall i,j \in [n], i\neq j.
\end{equation}
Then, \eqref{eq:constr} can be rewritten as
\begin{equation}\label{eq:fi}
\mathbf{p}_j =\sum_{i =1,i\neq j}^{n}{\mathbf{B}_{i,j}\mathbf{p}_{i,j}} \qquad \forall j\in [n].
\end{equation}
\if
where $\mathbf{f}_j(x_1,x_2,\dots ,x_n)=\sum_{i=1}^n \mathbf{f}_{i,j}(x_i)$, and
\begin{equation}\label{eq:p^j}
\mathbf{p}^j=(\mathbf{p}_{1,j},\mathbf{p}_{2,j},\dots, \mathbf{p}_{n,j}) \qquad \forall j \in [n].
\end{equation}
\fi
As a result, the parity symbols are the coded symbols 
from two sets of encoding matrices $\{\mathbf{A}_{i,j}\}_{i,j \in [n]}$ and $\{\mathbf{B}_{i,j}\}_{i,j \in [n]}$. 
This encoding aspect of the update model will be adopted 
in later sections.
Since the number of symbols passed from \eqref{eq:gi} to
\eqref{eq:fi} is $\displaystyle\sum_{i=1}^{n}{\sum_{j\in [n] \setminus \{i\}}{\gamma_{i,j}}}=n\gamma$, the average number of symbols
transmitted among all nodes
during the encoding process is equal to the update bandwidth $\gamma$.


\section{Necessary condition for the existence of an irregular array code}\label{sec:UBMD}

In this section, we provide a necessary condition for the parameters $\{p_{j}\}_{j\in [n]}$ and $\{\gamma_{i,j}\}_{i \neq j \in [n]}$ such that an $(n,k,\mathbf{m})$ irregular array code $\mathcal{C}$, where retrieval of data symbols
can only be guaranteed by any other $k$ columns but not by any other $k-1$ columns, exists (cf.~Theorem \ref{th:x3} and Corollary \ref{cor:1}). 
For simplicity, we use $H(\cdot)$ and $I(\cdot \ ; \ \cdot)$ to represent $H_{q}(\cdot)$ and $I_{q}(\cdot \ ; \ \cdot)$ in this section. 


Some notations used in the proofs below are first introduced (cf.~Table~\ref{tab:1}).
For a subset $\mathcal{E} \subset [n]$ with $|\mathcal{E}|=n-k$, the elements in $\mathcal{E}$ are denoted as $e_{i}$ with $i\in[n-k]$. Similarly, denote the elements in  $\bar{\mathcal{E}} \triangleq [n] \setminus \mathcal{E}$ as $\bar{e}_i$ with $i\in [k]$. Let $\mathbf{X}_{\mathcal{E}}\triangleq[\mathbf{x}_{e_{1}}^{\Tr} \dots \mathbf{x}_{e_{n-k}}^{\Tr}]^{\Tr}$ and $\mathbf{X}_{\bar{\mathcal{E}}}\triangleq[\mathbf{x}_{\bar{e}_{1}}^{\Tr} \dots \mathbf{x}_{\bar{e}_{k}}^{\Tr}]^{\Tr}$, and $\mathbf{C}_{\mathcal{E}}$, $\mathbf{C}_{\bar{\mathcal{E}}}$, $\mathbf{P}_{\mathcal{E}}$ and $\mathbf{P}_{\bar{\mathcal{E}}}$ are similarly defined. 
Equation \eqref{eq:constr} can then be rewritten using these notations as
\begin{equation}
\begin{aligned}
\mathbf{p}_{j}=&\sum_{i\in [n-k]}{\mathbf{M}_{e_{i},j}\mathbf{x}_{e_{i}}} + \sum_{i\in [k]}{\mathbf{M}_{\bar{e}_{i},j}\mathbf{x}_{\bar{e}_{i}}}.\\
\end{aligned}
\end{equation}
Thus, we can write
\begin{equation}\label{eq:x13}
\mathbf{P}_{\mathcal{\bar{E}}}=\begin{bmatrix}
\mathbf{p}_{\bar{e}_{1}}\\
\mathbf{p}_{\bar{e}_{2}}\\
\vdots\\
\mathbf{p}_{\bar{e}_{k}}\\
\end{bmatrix}=\begin{bmatrix}
\mathbf{M}_{e_{1},\bar{e}_{1}} & \dots & \mathbf{M}_{e_{n-k},\bar{e}_{1}} \\
\mathbf{M}_{e_{1},\bar{e}_{2}} & \dots & \mathbf{M}_{e_{n-k},\bar{e}_{2}} \\
\vdots &\ddots & \vdots \\
\mathbf{M}_{e_{1},\bar{e}_{k}} & \dots & \mathbf{M}_{e_{n-k},\bar{e}_{k}} \\
\end{bmatrix} \mathbf{X}_{\mathcal{E}}+\begin{bmatrix}
 \mathbf{M}_{\bar{e}_{1},\bar{e}_{1}} & \dots & \mathbf{M}_{\bar{e}_{k},\bar{e}_{1}}\\
 \mathbf{M}_{\bar{e}_{1},\bar{e}_{2}} & \dots & \mathbf{M}_{\bar{e}_{k},\bar{e}_{2}}\\
 \vdots & \ddots & \vdots\\
 \mathbf{M}_{\bar{e}_{1},\bar{e}_{k}} & \dots & \mathbf{M}_{\bar{e}_{k},\bar{e}_{k}}\\
\end{bmatrix} \mathbf{X}_{\bar{\mathcal{E}}}.
\end{equation}
Let 
\begin{equation}\label{eq:x14}
\mathbf{M}_{\mathcal{E}}\triangleq\begin{bmatrix}
\mathbf{M}_{e_{1},\bar{e}_{1}} & \dots & \mathbf{M}_{e_{n-k},\bar{e}_{1}} \\
\mathbf{M}_{e_{1},\bar{e}_{2}} & \dots & \mathbf{M}_{e_{n-k},\bar{e}_{2}} \\
\vdots &\ddots & \vdots \\
\mathbf{M}_{e_{1},\bar{e}_{k}} & \dots & \mathbf{M}_{e_{n-k},\bar{e}_{k}} \\
\end{bmatrix}, \qquad 
\mathbf{M}_{\bar{\mathcal{E}}}\triangleq\begin{bmatrix}
 \mathbf{M}_{\bar{e}_{1},\bar{e}_{1}} & \dots & \mathbf{M}_{\bar{e}_{k},\bar{e}_{1}}\\
 \mathbf{M}_{\bar{e}_{1},\bar{e}_{2}} & \dots & \mathbf{M}_{\bar{e}_{k},\bar{e}_{2}}\\
 \vdots & \ddots & \vdots\\
 \mathbf{M}_{\bar{e}_{1},\bar{e}_{k}} & \dots & \mathbf{M}_{\bar{e}_{k},\bar{e}_{k}}\\
\end{bmatrix}.
\end{equation}
Then, from~\eqref{eq:x13} and~\eqref{eq:x14}, we establish
\begin{equation}\label{eq:x15}
\mathbf{P}_{\mathcal{\bar{E}}}=\mathbf{M}_{\mathcal{E}} \mathbf{X}_{\mathcal{E}}+\mathbf{M}_{\bar{\mathcal{E}}} \mathbf{X}_{\bar{\mathcal{E}}}.
\end{equation}

In the following, we provide four lemmas that will be useful
in characterizing a necessary condition for the existence 
of an $(n,k,\mathbf{m})$ irregular array code
in Theorem~\ref{th:x2}.

\begin{lemma}\label{le:x1}
Given a matrix $\mathbf{A} \in \mathbb{F}_{q}^{a \times b}$ and a random column vector $\mathbf{b} \in \mathbb{F}_{q}^{b}$, we have $H(\mathbf{A}\mathbf{b}) \leq\rank(\mathbf{A})$.
\end{lemma}
\begin{proof} The lemma trivially holds when 
$\rank(\mathbf{A})=0$ or $\rank(\mathbf{A})=\row(\mathbf{A})$. Here, we provide the proof subject to $\row(\mathbf{A})>\rank(\mathbf{A})> 0$.
Given a matrix $\mathbf{A} \in \mathbb{F}_{q}^{a \times b}$, there is an invertible matrix $\mathbf{R}$ such that $\mathbf{R}\mathbf{A} = \tiny \begin{bmatrix}
\mathbf{A}'\\
[\mathbf{0}]
\end{bmatrix}$, where $[\mathbf{0}]$ is a $(\row(\mathbf{A})-\rank(\mathbf{A}) )\times\col(\mathbf{A})$ zero matrix. Then, given $\mathbf{A}'\mathbf{b}$, we can determine $\mathbf{A}\mathbf{b}$ via $\mathbf{A}\mathbf{b}=\mathbf{R}^{-1}\tiny \begin{bmatrix}
\mathbf{A}'\mathbf{b}\\
[\mathbf{0}]
\end{bmatrix}$, and vice versa. Thus, we have $H(\mathbf{A}\mathbf{b}|\mathbf{A}'\mathbf{b})=H(\mathbf{A}'\mathbf{b}|\mathbf{A}\mathbf{b})=0$. As $I(\mathbf{A}'\mathbf{b} ; \mathbf{A}\mathbf{b})=H(\mathbf{A}'\mathbf{b}) - H(\mathbf{A}'\mathbf{b}|\mathbf{A}\mathbf{b})=H(\mathbf{A}\mathbf{b}) - H(\mathbf{A}\mathbf{b}|\mathbf{A}'\mathbf{b})$, we conclude $H(\mathbf{A}\mathbf{b}) =H(\mathbf{A}'\mathbf{b}) \leq\row(\mathbf{A}'\mathbf{b})=\rank(\mathbf{A})$. This completes the proof.
\end{proof}

The next three lemmas associate $\mathbf{m}$, $\mathbf{p}$ and $\{\gamma_{i,j}=\rank(\mathbf{M}_{i,j})\}_{i,j\in [n]}$ through $\rank(\mathbf{M}_{\mathcal{E}})$.

\begin{lemma}\label{th:x2}
Given any $\mathcal{E}\subset [n]$ with $|\mathcal{E}|=n-k$, if each codeword $\mathbf{C}\in\mathcal{C}$ can be determined uniquely by $\mathbf{C}_{\bar{\mathcal{E}}}$, then we have $\sum_{i \in \mathcal{E}}{m_{i}} \leq\rank(\mathbf{M}_{\mathcal{E}})$.
\end{lemma}
\begin{proof}
If the knowledge of $\mathbf{C}_{\bar{\mathcal{E}}}$
can reconstruct the entire $\mathcal{C}$, then $H(\mathbf{X}_{\mathcal{E}} | \mathbf{C}_{\bar{\mathcal{E}}})=0$. Thus, we have
\begin{equation}\label{eq:23}
I(\mathbf{C}_{\bar{\mathcal{E}}} ; \mathbf{X}_{\mathcal{E}})=H(\mathbf{X}_{\mathcal{E}})-H(\mathbf{X}_{\mathcal{E}} | \mathbf{C}_{\bar{\mathcal{E}}})=H(\mathbf{X}_{\mathcal{E}})=\sum_{i \in \mathcal{E}}{m_i}.
\end{equation}
Since $I(\mathbf{X}_{\bar{\mathcal{E}}};\mathbf{X}_{\mathcal{E}})=0$
as indicated by \eqref{eq:ij}, we get
\begin{equation}\label{eq:ice}
I(\mathbf{C}_{\bar{\mathcal{E}}} ; \mathbf{X}_{\mathcal{E}})=I(\mathbf{X}_{\bar{\mathcal{E}}},\mathbf{P}_{\bar{\mathcal{E}}};\mathbf{X}_{\mathcal{E}})
=I(\mathbf{X}_{\bar{\mathcal{E}}};\mathbf{X}_{\mathcal{E}})+I(\mathbf{P}_{\bar{\mathcal{E}}};\mathbf{X}_{\mathcal{E}}|\mathbf{X}_{\bar{\mathcal{E}}})
=I(\mathbf{P}_{\bar{\mathcal{E}}};\mathbf{X}_{\mathcal{E}}|\mathbf{X}_{\bar{\mathcal{E}}}).
\end{equation}
We then obtain from~\eqref{eq:x15} and~\eqref{eq:ice} that
\begin{IEEEeqnarray*}{rCl}
I(\mathbf{P}_{\bar{\mathcal{E}}};\mathbf{X}_{\mathcal{E}}|\mathbf{X}_{\bar{\mathcal{E}}})&=&I(\mathbf{M}_{\mathcal{E}} \mathbf{X}_{\mathcal{E}}+\mathbf{M}_{\bar{\mathcal{E}}} \mathbf{X}_{\bar{\mathcal{E}}};\mathbf{X}_{\mathcal{E}}|\mathbf{X}_{\bar{\mathcal{E}}})\notag\\
&=&H(\mathbf{M}_{\mathcal{E}} \mathbf{X}_{\mathcal{E}}+\mathbf{M}_{\bar{\mathcal{E}}} \mathbf{X}_{\bar{\mathcal{E}}}|\mathbf{X}_{\bar{\mathcal{E}}})
- H(\mathbf{M}_{\mathcal{E}} \mathbf{X}_{\mathcal{E}}+\mathbf{M}_{\bar{\mathcal{E}}} \mathbf{X}_{\bar{\mathcal{E}}}|\mathbf{X}_{\bar{\mathcal{E}}},\mathbf{X}_{\mathcal{E}})\notag\\
&=&H(\mathbf{M}_{\mathcal{E}} \mathbf{X}_{\mathcal{E}}|\mathbf{X}_{\bar{\mathcal{E}}})\notag\\
&=&H(\mathbf{M}_{\mathcal{E}} \mathbf{X}_{\mathcal{E}}),\label{eq:x19}
\end{IEEEeqnarray*}
which, together with \eqref{eq:23},~\eqref{eq:ice} and Lemma~\ref{le:x1}, implies
\begin{equation*}\label{eq:x16}
\sum_{i \in \mathcal{E}}{m_i}=I(\mathbf{C}_{\bar{\mathcal{E}}} ; \mathbf{X}_{\mathcal{E}})=I(\mathbf{P}_{\bar{\mathcal{E}}};\mathbf{X}_{\mathcal{E}}|\mathbf{X}_{\bar{\mathcal{E}}})=H(\mathbf{M}_{\mathcal{E}} \mathbf{X}_{\mathcal{E}})
\leq\rank(\mathbf{M}_{\mathcal{E}}).
\end{equation*}
\end{proof}

\begin{lemma}\label{le:x2}
$\rank(\mathbf{M}_{\mathcal{E}}) \leq \sum_{j \in \bar{\mathcal{E}}}\min\{p_{j}, \sum_{i \in \mathcal{E}}{\gamma_{i,j}}\}$.
\end{lemma}
\begin{proof}
We first note from~\eqref{eq:x14} that
\begin{equation}\label{eq:x20}
\rank(\mathbf{M}_{\mathcal{E}}) \leq \sum_{j \in [k]}{\rank([
\mathbf{M}_{e_{1},\bar{e}_{j}} \dots \mathbf{M}_{e_{n-k},\bar{e}_{j}}
])}.
\end{equation}
Using $\gamma_{i,j}=\rank(\mathbf{M}_{i,j})$ from Theorem~\ref{th:0},
we obtain
\begin{equation}\label{eq:x21}
\rank([
\mathbf{M}_{e_{1},\bar{e}_{j}}  \dots  \mathbf{M}_{e_{n-k},\bar{e}_{j}}]) \leq \sum_{i \in [n-k]}{\rank(\mathbf{M}_{e_{i},\bar{e}_{j}})}=\sum_{i \in [n-k]}{\gamma_{e_{i}, \bar{e}_{j}}}.
\end{equation}
Next, we note that 
\begin{equation}\label{eq:x22}
\rank([
\mathbf{M}_{e_{1},\bar{e}_{j}}  \dots  \mathbf{M}_{e_{n-k},\bar{e}_{j}} ]) \leq\row([
\mathbf{M}_{e_{1},\bar{e}_{j}}  \dots  \mathbf{M}_{e_{n-k},\bar{e}_{j}} ])=p_{\bar{e}_{j}}.
\end{equation}
Combining~\eqref{eq:x21} and~\eqref{eq:x22} yields 
\begin{equation}\label{eq:x23}
\rank([
\mathbf{M}_{e_{1},\bar{e}_{j}}  \dots  \mathbf{M}_{e_{n-k},\bar{e}_{j}} ])  \leq \min\bigg\{p_{\bar{e}_{j}}, \sum_{i \in [n-k]}{\gamma_{e_{i}, \bar{e}_{j}}}\bigg\}.
\end{equation}
The validity of the lemma can thus be confirmed by~\eqref{eq:x20} and~\eqref{eq:x23}. 
\end{proof}

\begin{lemma}\label{le:x3}
$\rank(\mathbf{M}_{\mathcal{E}}) \leq \sum_{i \in \mathcal{E}}{\min\{m_{i}, \sum_{j \in \bar{\mathcal{E}}}{\gamma_{i,j}}\}}$.
\end{lemma}
\begin{proof}
The proof of this lemma is similar to that of Lemma~\ref{le:x2}.
First from~\eqref{eq:x14}, we establish
\begin{equation}\label{eq:x24}
\rank(\mathbf{M}_{\mathcal{E}})=\rank(\mathbf{M}_{\mathcal{E}}^{\Tr}) \leq \sum_{i \in [n-k]}{\rank([
\mathbf{M}_{e_{i},\bar{e}_{1}}^{\Tr}  \dots  \mathbf{M}_{e_{i},\bar{e}_{k}}^{\Tr} 
])}.
\end{equation}
In parallel to~\eqref{eq:x21} and~\eqref{eq:x22}, we next derive $\rank([
\mathbf{M}_{e_{i},\bar{e}_{1}}^{\Tr}  \dots  \mathbf{M}_{e_{i},\bar{e}_{k}}^{\Tr} 
]) \leq \sum_{j \in [k]}{\gamma_{e_{i},\bar{e}_{j}}}$ and $\rank([
\mathbf{M}_{e_{i},\bar{e}_{1}}^{\Tr}  \dots  \mathbf{M}_{e_{i},\bar{e}_{k}}^{\Tr} 
]) \leq\row([
\mathbf{M}_{e_{i},\bar{e}_{1}}^{\Tr}  \dots  \mathbf{M}_{e_{i},\bar{e}_{k}}^{\Tr} 
])=m_{e_{i}}$, which immediately gives 
\begin{equation}\label{eq:x25}
\rank([
\mathbf{M}_{e_{i},\bar{e}_{1}}^{\Tr}  \dots  \mathbf{M}_{e_{i},\bar{e}_{k}}^{\Tr} 
]) \leq \min\bigg\{m_{e_{i}}, \sum_{j \in [k]}{\gamma_{e_{i},\bar{e}_{j}}}\bigg\}.
\end{equation}
The lemma then follows from~\eqref{eq:x24} and~\eqref{eq:x25}. 
\end{proof}

After establishing the above four lemmas, we are now ready to prove the main result in this section.

\begin{theorem}\label{th:x3} 
Given any $\mathcal{E}\subset [n]$ with $|\mathcal{E}|=n-k$, if each codeword $\mathbf{C}\in\mathcal{C}$ can be determined uniquely by $\mathbf{C}_{\bar{\mathcal{E}}}$, then the following inequalities must hold:
\begin{equation}\label{eq:x26}
\sum_{j \in \bar{\mathcal{E}}}{\gamma_{i,j}} \geq m_{i} \qquad \forall i \in \mathcal{E},
\end{equation}
\begin{equation}\label{eq:x27}
\sum_{i \in \mathcal{E}}{m_{i}} \leq \sum_{j \in \bar{\mathcal{E}}}{\min\bigg\{p_{j}, \sum_{i \in \mathcal{E}}{\gamma_{i,j}}\bigg\}}.
\end{equation}
\end{theorem}
\begin{proof}
Inequality  
\eqref{eq:x27} is an immediate consequence of Lemmas~\ref{th:x2} and~\ref{le:x2}.

The inequality in~\eqref{eq:x26} can be proved by contradiction. 
Suppose   
$\sum_{j \in \bar{\mathcal{E}}}{\gamma_{u,j}} < m_{u}$ for some $u \in \mathcal{E}$.
Then, we can infer from Lemma~\ref{le:x3} that
\begin{equation}\label{eq:x28}
\rank(\mathbf{M}_{\mathcal{E}}) \leq
\sum_{i \in \mathcal{E} \setminus \{u\}}{\min\bigg\{m_{i}, \sum_{j \in \bar{\mathcal{E}}}{\gamma_{i,j}}\bigg\}}+
\sum_{j \in \bar{\mathcal{E}}}{\gamma_{u,j}}    < \sum_{i \in \mathcal{E} \setminus \{u\}}{m_{i}}+m_{u}   = \sum_{i \in \mathcal{E}}{m_{i}},
\end{equation}
which contradicts to Lemma~\ref{th:x2}.
Consequently, inequality \eqref{eq:x26} must hold for every
$i \in \mathcal{E}$. 
\end{proof}

For completeness, we conclude the section by reiterating the result in Theorem~\ref{th:x3} in the following corollary.

\begin{corollary}\label{cor:1}
An $(n,k,\mathbf{m})$ irregular array code $\mathcal{C}$ with construction matrices $\{\mathbf{M}_{i,j}\}_{i,j\in[n]}$ and 
numbers of parity symbols specified in 
$\mathbf{p}$ must fulfill \eqref{eq:x26} and~\eqref{eq:x27} for every $\mathcal{E} \subset [n]$ with $|\mathcal{E}|=n-k$.
\end{corollary}

\section{Lower bounds for code redundancy  and update bandwidth}\label{sec:MRMB}

In this section, three lower bounds will be established, 
which are lower bounds respectively for code redundancy and update bandwidth, and a lower bound for code redundancy subject to the minimum update bandwidth.
Their achievability by explicit constructions 
of irregular array codes under $k \mid m_{i}$ for all $i \in [n]$
will be shown in Section~\ref{sec:code}.
Without loss of generality, we assume in this section that
\begin{equation}\label{eq:mi}
m_1\geq m_2 \geq \dots \geq m_n\geq 0.
\end{equation}

\subsection{Minimization of code redundancy}

Theorem~\ref{th:x3} indicates that a lower bound for the code redundancy of an $(n,k,\mathbf{m})$ irregular array code can be obtained by solving the linear programming problem below.
\begin{lp}\label{lp:LP1} To minimize 
$R=\sum_{i=1}^{n}{p_i}$\,, subject to~\eqref{eq:x26} and~\eqref{eq:x27}
among all $\mathcal{E} \subset [n]$ with $|\mathcal{E}|=n-k$.
\end{lp}

Since the object function of Linear Programming \ref{lp:LP1} is only a function of $\mathbf{p}$,
a code redundancy $R$ is attainable due to a choice of $\mathbf{p}$, 
if there exists a set of corresponding $\{\gamma_{i,j}\}_{i\neq j\in[n]}$ that can validate both 
\eqref{eq:x26} and~\eqref{eq:x27}.
A valid selection of such $\{\gamma_{i,j}\}_{i\neq j\in[n]}$ for a given $\mathbf{p}$
is to persistently increase $\{\gamma_{i,j}\}_{i\neq j\in[n]}$ until 
both \eqref{eq:x26} and 
\begin{equation}
p_{j}\leq \sum_{i \in \mathcal{E}}{\gamma_{i,j}}\qquad \forall j \in \bar{\mathcal{E}}
\end{equation}
are satisfied for arbitrary choice of 
$\mathcal{E}\subset [n]$ with $|\mathcal{E}|=n-k$.
As a result, we can disregard \eqref{eq:x26} and reduce
\eqref{eq:x27}
to 
\begin{equation}\label{eq:x30}
\sum_{i \in \mathcal{E}}{m_{i}}\leq \sum_{j \in \bar{\mathcal{E}}}{p_{j}},
\end{equation}
leading to a new linear programming setup as follows.

\begin{lp}\label{lp:LP2}
To minimize $R=\sum_{i=1}^{n}{p_i}$, subject to \eqref{eq:x30}
among all $\mathcal{E} \subset [n]$ with $|\mathcal{E}|=n-k$.
\end{lp}

\begin{lemma}\label{le:x4}
Linear Programming~\ref{lp:LP1} is equivalent to Linear Programming~\ref{lp:LP2}.
\end{lemma}
\begin{proof}
It is obvious that all minimizers of Linear Programming~\ref{lp:LP1} 
satisfy the constraint in Linear Programming~\ref{lp:LP2}.
On the contrary, given a minimizer $\mathbf{p}$ of Linear Programming~\ref{lp:LP2}, we can assign $\gamma_{i,j}=\max\{m_{i},p_{j}\}$ 
to satisfy the constraints in Linear Programming~\ref{lp:LP1}. Thus, Linear Programming~\ref{lp:LP1} and Linear Programming~\ref{lp:LP2}
are equivalent.
\end{proof}

{\bf Remark.}
Note that Linear Programming~\ref{lp:LP2}, which was first given in \cite{tosato2014}, is not related to the update bandwidth $\gamma$ of an irregular array code, while the proposed setup in Linear Programming~\ref{lp:LP1} is.
Thus, the latter setup can be used to determine an 
irregular array code of update-bandwidth efficiency 
 by replacing the object function $R$ with update bandwidth $\gamma$. However, for the minimization of code redundancy, the two linear programming settings are equivalent as confirmed in Lemma~\ref{le:x4}.

To solve Linear Programming~\ref{lp:LP2}, Tosato and Sandell \cite{tosato2014} introduced 
 a water level parameter $\mu$, defined as
\begin{equation}\label{eq:x29}
\mu=\max\left\{m_{n-k}, \left\lceil \frac{B}{k} \right\rceil \right\},
\end{equation}
where $B \triangleq \sum_{i \in [n]}{m_{i}}$ is the total number of data symbols,
and by following the assumption
in \eqref{eq:mi}, $m_{n-k}$ is the $(n-k)$-th largest element in vector $\mathbf{m}$.
It was shown in \cite{tosato2014} that the minimum code redundancy equals
\begin{equation}\label{eq:39}
R_{\min}=\sum_{i=1}^{n-k}([\mu-m_i]_{+}+m_i),
\end{equation} 
which can only be achieved by those $\mathbf{p}$'s satisfying
\begin{equation}\label{eq:47}
\begin{cases}
p_i=[\mu-m_i]_{+}\quad\text{for } 1\leq i \leq n-k,\\
p_i\leq[\mu-m_i]_+\quad\text{for }n-k<i\leq n,\\
\sum_{i=n-k+1}^{n}{p_i}=\sum_{i=1}^{n-k}{m_i},
\end{cases}
\end{equation}
where $[x]_{+}\triangleq\max\{0,x\}$. 
The class of $(n,k,\mathbf{m})$ irregular array codes conforming to \eqref{eq:47} is called irregular MDS array codes~\cite{tosato2014}. 
In particular, when $k \mid B$ and $m_1\leq {B}/{k}$, \eqref{eq:39} and \eqref{eq:47}  
can be respectively reduced to 
\begin{equation}\label{eq:37}
R_{\min}=\frac{(n-k)}{k}B,
\end{equation}
and
\begin{equation}\label{eq:37a}
p_{i}=\frac{B}{k}-m_{i}\qquad\forall i\in[n].
\end{equation}

\subsection{Minimization of update bandwidth}

We now turn to the determination of the minimum update bandwidth.
As similar to Linear Programming \ref{lp:LP1}, a lower bound 
for the update bandwidth of an $(n,k,\mathbf{m})$ irregular array code can be obtained using the linear programming below.

\begin{lp}\label{lp:LP4}
To minimize $\gamma=\frac{1}{n}\sum_{i=1}^{n}{\sum_{j\in [n] \setminus \{i\}}{\gamma_{i,j}}}$, subject to \eqref{eq:x26} and~\eqref{eq:x27}
among all $\mathcal{E} \subset [n]$ with $|\mathcal{E}|=n-k$.
\end{lp}
Since $\mathbf{p}$ is not used in the above object function, 
a choice of $\{\gamma_{i,j}\}_{i\neq j\in[n]}$ is feasible for the minimization of $\gamma$ 
as long as there is a corresponding $\mathbf{p}$ that validates both 
\eqref{eq:x26} and~\eqref{eq:x27}.
A valid selection of such $\mathbf{p}$ for given $\{\gamma_{i,j}\}_{i\neq j\in[n]}$
is to set $p_j=\sum_{i \in [n]}{\gamma_{i,j}}$, which reduces
\eqref{eq:x27} to a consequence of \eqref{eq:x26}, i.e., 
\begin{equation}\label{eq:x36}
\sum_{i \in \mathcal{E}}{m_{i}}\leq \sum_{j \in \bar{\mathcal{E}}}\sum_{i \in \mathcal{E}}{\gamma_{i,j}}.
\end{equation}
As a result, by following an analogous proof to that used in Lemma~\ref{le:x4},
Linear Programming \ref{lp:LP4} can also be  
solved through the following equivalent setup.

\begin{lp}\label{lp:LP4'}
To minimize $\gamma=\frac{1}{n}\sum_{i=1}^{n}{\sum_{j\in [n] \setminus \{i\}}{\gamma_{i,j}}}$ subject to \eqref{eq:x26} among all $\mathcal{E} \subset [n]$ with $|\mathcal{E}|=n-k$.
\end{lp}

The solution of Linear Programming \ref{lp:LP4'} is then given 
in the following theorem.

\begin{theorem}(Minimum update bandwidth)\label{the:2}
The minimum update bandwidth determined through Linear Programming \ref{lp:LP4'} is given by 
\begin{equation}\label{gamma_min}
\gamma_{\min} = \frac{B}{n}+ \frac{(n-k-1)}{n}\sum_{i \in [n]}{\left\lceil \frac{m_{i}}{k} \right\rceil}.
\end{equation}
Under $k<n-1$, the minimum update bandwidth can only be achieved by the assignment that satisfies 
for every $i\in[n]$, 
\begin{equation}\label{eq:xxx''}
\sum_{u\in[w_i]}\gamma_{i, j_u(i)}=w_i\left\lfloor \frac{m_{i}}{k} \right\rfloor,
\text{ and }\gamma_{i, j_u(i)}=
\left\lceil \frac{m_{i}}{k} \right\rceil\text{for }u\in[n-1]\setminus
[w_i],
\end{equation}
where $w_i\triangleq k\left\lceil\frac{m_i}{k}\right\rceil-m_i<k$, 
and for notational convenience, we let 
the indices $j_{1}(i), j_{2}(i), \ldots, j_{n-1}(i)$ be a permutation of $[n]\setminus\{i\}$ such that
\begin{equation}
\label{eq:sorting}	
0 \leq \gamma_{i, j_{1}(i)} \leq \dots \leq \gamma_{i, j_{n-1}(i)}\quad\mbox{for }i \in [n].
\end{equation}
When $k=n-1$, any $\{\gamma_{i,j}\}_{i\neq j\in[n]}$ that achieves $\gamma_{\min}$ must satisfy 
\begin{equation}\label{eq:x36'}
\sum_{j \in [n] \setminus \{i\}}{\gamma_{i,j}} = m_{i} \quad \forall  i \in [n]. 
\end{equation}

\end{theorem}
\begin{proof}
The proof is divided into four steps. First, we show all choices of $\{\gamma_{i,j}\}_{i\neq j\in[n]}$
satisfying \eqref{eq:x26} yield an update bandwidth no less than the $\gamma_{\min}$ given in \eqref{gamma_min}. Second, we verify
 \eqref{eq:xxx''} can achieve $\gamma_{\min}$.
 Third, we prove that \eqref{eq:xxx''}
  is the only assignment that can achieve $\gamma_{\min}$ under $k<n-1$. Last, we complete the proof by considering separately the situation of $k=n-1$. 

\begin{description}
\item[{\it Step 1.}]\ \ Fix a set of $\{\gamma_{i,j}\}_{i\neq j\in[n]}$ satisfying \eqref{eq:x26}. 
Since \eqref{eq:x26} holds for arbitrary $\mathcal{E}$, we can let
$\bar{\mathcal{E}}=\{j_1(i),\ldots,j_k(i)\}$ and obtain
\begin{equation}\label{eq:27}
\sum_{u\in[k]} \gamma_{i,j_u(i)} \geq m_{i}.
\end{equation}
Noting that $\{\gamma_{i,j_u(i)}\}_{u\in[n-1]}$ is in ascending order (cf.~\eqref{eq:sorting}), 
and that $\gamma_{i,j}$ is a non-negative integer, we obtain from \eqref{eq:27} that
\begin{equation}\label{eq:28''}
\gamma_{i, j_k(i)} \geq \left\lceil \frac{m_{i}}{k} \right\rceil.
\end{equation}
We continue to derive
\begin{equation}\label{eq:57'}
\begin{aligned}
&\sum_{j\in [n]\setminus\{i\}}{\gamma_{i,j}}= \sum_{u\in[k]}{\gamma_{i, j_u(i)}}+ \sum_{u\in[n-1]\setminus[k]}{\gamma_{i, j_u(i)}} \geq \sum_{u\in[k]}{\gamma_{i, 
j_u(i)}}+ (n-k-1)\gamma_{i, j_k(i)}.
\end{aligned}
\end{equation}
Combining \eqref{eq:27}, \eqref{eq:28''} and \eqref{eq:57'} gives 
\begin{equation}\label{eq:57''}
\begin{aligned}
&\sum_{j\in [n]\setminus\{i\}}{\gamma_{i,j}} \geq m_i+(n-k-1)\left\lceil \frac{m_{i}}{k} \right\rceil,
\end{aligned}
\end{equation}
which implies 
\begin{equation}
\begin{aligned}
&\gamma= \frac{1}{n}\sum_{i\in [n]}\sum_{j\in [n]\setminus\{i\}}{\gamma_{i,j}}\geq \frac{B}{n}+ \frac{(n-k-1)}{n}\sum_{i \in [n]}{\left\lceil \frac{m_{i}}{k}  \right\rceil}  = \gamma_{\min}.
\end{aligned}
\end{equation}
\item[{\it Step 2.}]\ \ Next, we confirm \eqref{eq:xxx''} is a valid choice
of $\{\gamma_{i,j}\}_{i\neq j \in [n]}$ that achieves 
$\gamma_{\min}$. 
The validity of \eqref{eq:x26} can be confirmed by 
\begin{equation}\label{wirange}
w_i=k\left\lceil\frac{m_i}{k}\right\rceil-m_i<k\left(\frac{m_i}{k}+1\right)-m_i=k
\end{equation}
and
\begin{IEEEeqnarray}{rCl}
\sum_{j \in \bar{\mathcal{E}}}{\gamma_{i,j}}\geq \sum_{i\in[k]}{\gamma_{i, j_u(i)}}
&=&w_{i}\left\lfloor \frac{m_{i}}{k} \right\rfloor+(k-w_{i})\left\lceil \frac{m_{i}}{k} \right\rceil\\
&=&
k\left\lceil \frac{m_i}{k} \right\rceil-\left(k\left\lceil \frac{m_i}{k} \right\rceil - m_i\right)\left(\left\lceil \frac{m_i}{k} \right\rceil-\left\lfloor \frac{m_i}{k} \right\rfloor\right)\label{eq:52}\\
&=&\begin{cases}
k\left\lceil \frac{m_i}{k} \right\rceil-0,&k \mid m_{i}\\
k\left\lceil \frac{m_i}{k} \right\rceil-\left(k\left\lceil \frac{m_i}{k} \right\rceil - m_i\right),&k\nmid m_i
\end{cases}\\
&=&m_i. \label{eq:x39'}
\end{IEEEeqnarray}
Hence, we can derive based on \eqref{eq:sorting} and \eqref{eq:x39'} that 
\begin{equation}
\begin{aligned}
&\sum_{j\in [n]\setminus\{i\}}{\gamma_{i,j}}=\sum_{u\in[k]}{\gamma_{i, j_u(i)}}+ \sum_{u\in[n-1]\setminus[k]}{\gamma_{i, j_u(i)}} = m_{i}+(n-k-1)\left\lceil \frac{m_{i}}{k} \right\rceil,
\end{aligned}
\end{equation}
which immediately gives
\begin{equation}
\begin{aligned}
&\gamma= \frac{1}{n}\sum_{i\in [n]}\sum_{j\in [n]\setminus\{i\}}{\gamma_{i,j}} = \frac{B}{n}+ \frac{(n-k-1)}{n}\sum_{i\in [n]}{\left\lceil \frac{m_{i}}{k} \right\rceil} = \gamma_{\min}.
\end{aligned}
\end{equation}

\item[{\it Step 3}.]\ \ 
It remains to show by contradiction that no other assignment of $\{\gamma_{i,j}\}_{i\neq j\in[n]}$ can achieve $\gamma_{\min}$. The task will be done under $k<n-1$ in this step. The situation of $k=n-1$ will be separately considered in the next step.

\hspace*{6mm}Suppose $\{\gamma'_{i,j}\}_{i \neq j \in [n]}$ also achieves $\gamma_{\min}$. We then differentiate among four cases.
\begin{description}
\item[{\it Case 1:}]\quad If there is $i'\in [n]$ such that $\gamma'_{i', j_{w_{i'}+1}(i')} > \left\lceil \frac{m_{i'}}{k} \right\rceil$, then we obtain from \eqref{eq:sorting} and \eqref{wirange} that $\gamma'_{i', j_{k}(i')} > \left\lceil \frac{m_{i'}}{k} \right\rceil$,
which together with \eqref{eq:27} implies
\begin{equation}\label{eq:55o}
\sum_{j \in [n] \setminus \{i'\}}{\gamma'_{i',j}}= \sum_{u \in [k]}{\gamma'_{i',j_u(i')}} + \sum_{u \in [n-1]\setminus [k]}{\gamma'_{i',j_u(i')}} > m_{i'} + (n-k-1)
\left\lceil\frac{m_{i'}}{k}\right\rceil.
\end{equation}
Because $\{\gamma'_{i,j}\}_{i \neq j \in [n]}$ fulfills \eqref{eq:x26}
and hence validates \eqref{eq:57''}, we have 
\begin{IEEEeqnarray}{rCl}\label{eq:x49}
\gamma&=& \frac{1}{n}\sum_{i\in [n]}\sum_{j\in [n]\setminus \{i\}}{\gamma'_{i,j}}\\
&=& \frac{1}{n}\left(\sum_{j \in [n] \setminus \{i'\}}{\gamma'_{i',j}}+\sum_{i\in [n]\setminus \{i'\}}\sum_{j\in [n]\setminus \{i\}}{\gamma'_{i,j}}\right) \\
&>&\frac{1}{n}\bigg(\left(m_{i'}+(n-k-1)\left\lceil \frac{m_{i'}}{k} \right\rceil\right)+\sum_{i\in [n]\setminus \{i'\}}
\left(m_i+(n-k-1)\left\lceil \frac{m_{i}}{k} \right\rceil\bigg)
\right)\label{eq:contra}\\
&=& \gamma_{\min},
\end{IEEEeqnarray}
contradicting to the assumption of $\{\gamma'_{i,j}\}_{i \neq j \in [n]}$ achieving $\gamma_{\min}$.

\item[{\it Case 2:}]\quad If there is $i' \in [n]$ such that $\gamma'_{i', j_{w_{i'}+1}({i'})} < \left\lceil \frac{m_{i'}}{k} \right\rceil$ and $w_{i'}+1=k$, then 
 we can infer from \eqref{eq:sorting} that
\begin{equation}
\gamma'_{i', j_{w_{i'}}({i'})} \leq \left\lfloor \frac{m_{i'}}{k} \right\rfloor,
\end{equation}
and hence 
\begin{IEEEeqnarray}{rCl}
\sum_{u\in[k]}\gamma'_{i', j_u(i')}&=&
\sum_{u\in[w_{i'}]}\gamma'_{i', j_u(i')}+\gamma'_{i', j_{w_{i'}+1}(i')}
+\underbrace{\sum_{u\in[k]\setminus[w_{i'}+1]}\gamma'_{i', j_u(i')}}_{=0}\\
&<&w_{i'}\left\lfloor \frac{m_{i'}}{k} \right\rfloor
+\left\lceil \frac{m_{i'}}{k} \right\rceil+\underbrace{(k-w_{i'}-1)\left\lceil \frac{m_{i'}}{k}\right\rceil}_{=0}=m_{i'}\label{eq:case2-1}
\end{IEEEeqnarray}
where the first strict inequality in \eqref{eq:case2-1} is due to $\gamma'_{i', j_{w_{i'}+1}({i'})} < \left\lceil \frac{m_{i'}}{k} \right\rceil$, and the  last equality follows from 
a similar derivation to \eqref{eq:x39'}. The inequality \eqref{eq:case2-1}  then contradicts to \eqref{eq:27}.

\item[{\it Case 3:}]\quad If there is $i' \in [n]$ such that $\gamma'_{i', j_{w_{i'}+1}({i'})} < \left\lceil \frac{m_{i'}}{k} \right\rceil$ and $w_{i'}+1<k$, then we must have 
$\gamma'_{i', j_{k}({i'})} >
\left\lceil \frac{m_{i'}}{k} \right\rceil$.
This is because in case $\gamma'_{i', j_{k}({i'})} \leq 
\left\lceil \frac{m_{i'}}{k} \right\rceil$ under $\gamma'_{i', j_{w_{i'}+1}({i'})} < \left\lceil \frac{m_{i'}}{k} \right\rceil$ and $w_{i'}+1<k$, we can obtain from \eqref{eq:sorting} that
\begin{equation}
\gamma'_{i', j_{w_{i'}}({i'})} \leq \left\lfloor \frac{m_{i'}}{k} \right\rfloor,
\end{equation}
and hence a similar derivation to \eqref{eq:case2-1} gives
\begin{IEEEeqnarray}{rCl}
\sum_{u\in[k]}\gamma'_{i', j_u(i')}&=&
\sum_{u\in[w_{i'}]}\gamma'_{i', j_u(i')}+\gamma'_{i', j_{w_{i'}+1}(i')}
+\sum_{u\in[k]\setminus[w_{i'}+1]}\gamma'_{i', j_u(i')}\\
&<&w_{i'}\left\lfloor \frac{m_{i'}}{k} \right\rfloor
+\left\lceil \frac{m_{i'}}{k} \right\rceil+(k-w_{i'}-1)\left\lceil \frac{m_{i'}}{k}\right\rceil=m_{i'}.\label{eq:case2}
\end{IEEEeqnarray}
The inequality \eqref{eq:case2}  then contradicts to \eqref{eq:27},
and therefore $\gamma'_{i', j_{k}({i'})} >
\left\lceil \frac{m_{i'}}{k} \right\rceil$.
We continue to derive based on \eqref{eq:27} that 
\begin{equation}
\sum_{j \in [n] \setminus \{i'\}}{\gamma'_{i',j}}= \sum_{u \in [k]}{\gamma'_{i',j_u(i')}} + \sum_{u \in [n-1]\setminus [k]}{\gamma'_{i',j_u(i')}} > m_{i'} + (n-k-1)
\left\lceil\frac{m_{i'}}{k}\right\rceil,
\end{equation}
based on which the same contradiction as \eqref{eq:contra} can be resulted.

\item[{\it Case 4:}]\quad The previous three cases indicate that 
$\gamma'_{i', j_{w_{i'}+1}({i'})} = \left\lceil\frac{m_{i'}}{k}\right\rceil$ for all $i'\in[n]$.
Now if there is $i'\in[n]$ and $w_{i'}< u'\leq n-1$ such that $\gamma'_{i', j_{u'}({i'})}<\gamma'_{i', j_{u'+1}(i')}$, then we again use \eqref{eq:27} to obtain
\begin{IEEEeqnarray}{rCl}
\sum_{j \in [n] \setminus \{i'\}}{\gamma'_{i',j}}&=& \sum_{u\in [k]}{\gamma'_{i',j_u(i')}} + \sum_{u \in [n-1]\setminus [k]}{\gamma'_{i',j_u(i')}}\\
& >& m_{i'} + (n-k-1)\left\lceil\frac{m_{i'}}{k}\right\rceil,
\end{IEEEeqnarray}
based on which the same contradiction as \eqref{eq:contra} can, again, be resulted.
\end{description}
The above four cases conclude that $\gamma'_{i, j_{u}({i})}=\left\lceil\frac{m_i}k\right\rceil$ for $u\in[n-1]\setminus[w_i]$ and $i\in[n]$.
Finally, \eqref{eq:57''} implies
\begin{equation}\label{eq:final}
\begin{aligned}
&\sum_{j\in [n]\setminus\{i\}}{\gamma'_{i,j}} 
= \sum_{u\in[w_i]}{\gamma'_{i,j_u(i)}}+(n-w_i-1)\left\lceil\frac{m_i}{k}\right\rceil \geq m_i+(n-k-1)\left\lceil \frac{m_{i}}{k} \right\rceil.
\end{aligned}
\end{equation}
Since the sum of the left-hand-side of \eqref{eq:final} 
is equal to the sum of the right-hand-side of \eqref{eq:final}, which is exactly $\gamma_{\min}$, 
we must have
\begin{equation}
\begin{aligned}
&\sum_{u\in[w_i]}{\gamma'_{i,j_u(i)}}+(n-w_i-1)\left\lceil\frac{m_i}{k}\right\rceil 
=m_i+(n-k-1)\left\lceil \frac{m_{i}}{k} \right\rceil,
\end{aligned}
\end{equation}
which in turn gives
\begin{IEEEeqnarray}{rCl}
\sum_{u\in[w_i]}{\gamma'_{i,j_u(i)}}&=&
m_i+(w_i-k)\left\lceil \frac{m_{i}}{k} \right\rceil=w_i\left\lfloor \frac{m_{i}}{k}\right\rfloor,
\end{IEEEeqnarray}
where the last equality can be confirmed similarly as \eqref{eq:x39'}.

\item[{\it Step 4.}]\ \  Last, we prove \eqref{eq:x36'}. Note that the proofs in Steps 1 and 2 remain valid under $k=n-1$, but 
some derivations in Step 3, e.g., \eqref{eq:55o}, may not be applied when $k=n-1$.\footnote{Note that under $k=n-1$, \eqref{eq:xxx''} is no longer the only assignment that achieves $\gamma_{\min}$. For example, 
for an $(3,2,\mathbf{m}=[5\ 5\ 5]^{\Tr})$ irregular array code, the assignment 
of \eqref{eq:xxx''} gives $\gamma_{i, j_1(i)}=2$ and $\gamma_{i, j_2(i)}=3$ for $i\in[3]$, but 
\begin{equation}
\gamma_{i,j}=\begin{cases}
5,&(i,j)\in\{(1,3),(2,3),(3,2)\}\\
0,&(i,j)\in\{(1,2),(2,1),(3,1)
\end{cases}
\end{equation}
can also achieve $\gamma_{\min}=5$.
This justifies our separate consideration of the case of $k=n-1$.
}
In fact, when $k=n-1$, a larger class of assignments on $\{\gamma_{i,j}\}_{i\neq j\in[n]}$ can achieve $\gamma_{\min}$. We show \eqref{eq:x36'} by contradiction.
Suppose $\{\gamma'_{i,j}\}_{i \neq j \in [n]}$ achieves $\gamma_{\min}$ but satisfies $\sum_{j \in [n] \setminus \{i'\}}{\gamma_{i',j}'} > m_{i'}$ for some $i' \in [n]$.
Then, 
\begin{IEEEeqnarray}{rCl}
\gamma&=&\frac{1}{n}\left(\sum_{i \in [n] \setminus \{i'\}}\sum_{j \in [n] \setminus \{i\}}{\gamma_{i,j}'}+\sum_{j \in [n] \setminus \{i'\}}{\gamma_{i',j}'}\right)\\
&>&\frac{1}{n}\left(\sum_{i \in [n] \setminus \{i'\}}m_i+m_{i'}\right)\\
&=& \frac{B}{n}=\gamma_{\min},
\end{IEEEeqnarray}
which leads to a contradiction. 
\end{description}
\end{proof} 
 
\subsection{Determination of the smallest code redundancy
subject to $\gamma=\gamma_{\min}$}\label{sec:4c}

In Theorem \ref{the:2}, the class of optimal $\{\gamma_{i,j}\}_{i \neq j \in [n]}$ that achieve $\gamma_{\min}$ is also determined. In particular, 
when $k<n-1$ and $k\mid m_i$ for every $i$,
we have that 
\begin{equation}\label{eq:sl1}
\gamma_{i,j}=\frac{m_i}k\quad\forall i\neq j\in[n]
\end{equation}
uniquely achieves $\gamma_{\min}$.
This facilitates our finding the smallest code redundancy attainable subject to $\gamma=\gamma_{\min}$ as formulated in Linear Programming \ref{lp:LP5} below.

\begin{lp}\label{lp:LP5}
To minimize $R=\sum_{i=1}^{n}{p_i}$ subject to~\eqref{eq:x27} and~\eqref{eq:xxx''} among all $\mathcal{E} \subset [n]$ with $|\mathcal{E}|=n-k$,
provided $1\leq k < n-1$ and $k \mid m_{i}$ for all $i \in [n]$.
\end{lp}

\begin{theorem}\label{th:x4}
The solution of Linear Programming~\ref{lp:LP5} is given by 
\begin{equation}\label{eq:Rsma}
R_{\sma} \triangleq \frac{(n-1)}{k}\sum_{i=1}^{n-k}{m_{i}}+\frac{(n-k)}{k}m_{n-k+1},
\end{equation}
where by following the assumption
in \eqref{eq:mi}, $m_{i}$ is the $i$-th largest element in vector $\mathbf{m}$.
The smallest code redundancy subject to $\gamma=\gamma_{\min}$
is uniquely achieved by
\begin{equation}\label{eq:x54}
p_{j}= \begin{cases}
\frac 1{k}{\sum_{i=1}^{n-k+1}{m_{i}}-m_j} &\text{for } 1\leq j \leq n-k,\\
\frac 1k\sum_{i=1}^{n-k}{m_{i}} &\text{for } n-k < j \leq n.
\end{cases}
\end{equation}
\end{theorem}
\begin{proof}
We first prove by contradiction that
\begin{equation}\label{eq:x57}
p_{j}\geq \sum_{i \in \mathcal{E}}{\gamma_{i,j}}=\frac 1k\sum_{i \in \mathcal{E}}{m_{i}} \qquad \forall \mathcal{E}\text{ and }\forall j \in \bar{\mathcal{E}}.
\end{equation}
Suppose there are $\mathcal{E} \subset [n]$ with $|\mathcal{E}| =n-k$ and $j' \in \bar{\mathcal{E}}$ such that
\begin{equation}
p_{j'} < \sum_{i \in \mathcal{E}}{\gamma_{i,j'}} = \frac 1k\sum_{i \in \mathcal{E}}{m_{i}}.
\end{equation}
Then,~\eqref{eq:x27} results in a contradiction as follows:
\begin{IEEEeqnarray}{rCl}
\sum_{i \in \mathcal{E}}{m_{i}} &\leq& \min\left\{p_{j'}, \sum_{i \in \mathcal{E}}{\gamma_{i,j'}}\right\} + \sum_{j \in \bar{\mathcal{E}} \setminus \{j'\}}{\min\left\{p_{j}, \sum_{i \in \mathcal{E}}{\gamma_{i,j}}\right\}} \\
&<& \sum_{j \in \bar{\mathcal{E}}}\sum_{i \in \mathcal{E}}{\gamma_{i,j}}\\
&=&\sum_{j \in \bar{\mathcal{E}}}\sum_{i \in \mathcal{E}}{\frac{m_i}k}
=\sum_{i \in \mathcal{E}}{{m_i}},\label{eq:78}
\end{IEEEeqnarray}
where $\eqref{eq:78}$ follows from \eqref{eq:sl1}.
Thus, \eqref{eq:x57} holds for arbitrary $\mathcal{E}\subset[n]\setminus\{j\}$.
As a result, we have
\begin{equation}\label{eq:x56}
p_{j}\geq \max_{\mathcal{E}\subset[n]\setminus\{j\}:|\mathcal{E}|=n-k}\frac 1k\sum_{i \in \mathcal{E}}{m_{i}}
=\begin{cases}
\frac 1{k}\left({\sum_{i=1}^{n-k+1}{m_{i}}-m_j}\right) &\text{for } 1\leq j \leq n-k,\\
\frac 1{k}{\sum_{i=1}^{n-k}{m_{i}}} &\text{for } n-k < j \leq n,
\end{cases}
\end{equation}
which implies
\begin{equation}\label{eq:x58}
R =\sum_{j=1}^{n}{p_j}\geq \frac{(n-1)}{k}\sum_{i=1}^{n-k}{m_{i}}+\frac{(n-k)}{k}m_{n-k+1} = R_{\sma}.
\end{equation}
Since any $\{p_j\}_{j\in[n]}$ that satisfies \eqref{eq:x56} with strict inequality
for some $j\in[n]$ cannot achieve $R_{\sma}$,
the smallest code redundancy subject to $\gamma=\gamma_{\min}$
is uniquely achieved by the one that fulfills \eqref{eq:x56} with equality.
\end{proof}

The contradiction proof 
in \eqref{eq:78} requires
$\sum_{j\in\bar{\mathcal{E}}}\gamma_{i,j}=m_i$, which 
is guaranteed by \eqref{eq:xxx''} when $k\mid m_i$ for all $i\in[n]$.
However, without $k\mid m_i$ for all $i\in[n]$, the 
$\sum_{j\in\bar{\mathcal{E}}}\gamma_{i,j}$
 in \eqref{eq:xxx''} may not achieve $m_i$
but generally lies between $m_i$ and $k\left\lceil\frac{m_i}k\right\rceil$.
Our preliminary study indicates that 
the general formula of $R_{\sma}$ for arbitrary $k<n-1$ and arbitrary $\mathbf{m}$ does not seem to have a simple expression but depends on the pattern of $\mathbf{w}=[w_1\ w_2\ \cdots w_n]^{\Tr}$. 
Theorem~\ref{th:xx5} only deals with $\mathbf{w}=[0\ 0\ \cdots\ 0]^{\Tr}$.
The establishment of the smallest code redundancy
for cases that allow $k \nmid m_{i}$ is left as a future research.

Surprisingly, in the particular case of $k=n-1$, we found 
$R_{\sma}=R_{\min}$  due to the fact that 
$\sum_{j\in\bar{\mathcal{E}}}\gamma_{i,j}=m_i$ is guaranteed by  
\eqref{eq:x36'}.

\begin{lp}\label{lp:LP6}
To minimize $R=\sum_{i=1}^{n}{p_i}$ subject to~\eqref{eq:x27} and~\eqref{eq:x36'} 
among all $\mathcal{E} \subset [n]$ with $|\mathcal{E}|=n-k$,
provided $k = n-1$.
\end{lp}
\begin{theorem}\label{th:xx5}
The solution of Linear Programming~\ref{lp:LP6} is given by the $R_{\min}$ in \eqref{eq:39}, which can only be achieved by those $\mathbf{p}$'s 
satisfying \eqref{eq:47}.
\end{theorem}
\begin{proof}
It suffices to prove that Linear Programming 
\ref{lp:LP2} and Linear Programming \ref{lp:LP6} are equivalent under $n-k=1$.

We first note that under $n-k=1$, all feasible $\mathbf{p}$ and $\{\gamma_{i,j}\}_{i\neq j\in[n]}$
satisfying \eqref{eq:x27} and \eqref{eq:x36'}, i.e., 
\begin{equation}\label{eq:82}
m_i=\sum_{j \in [n] \setminus \{i\}}{\gamma_{i,j}} \leq \sum_{j \in [n]\setminus\{i\}}{\min\{p_{j}, \gamma_{i,j}\}}
\quad\forall i\in[n],
\end{equation}
must validate \eqref{eq:x30}, i.e., 
\begin{equation}\label{eq:83}
m_{i}\leq \sum_{j \in [n]\setminus\{i\}}{p_{j}}\quad\forall i\in[n].
\end{equation}
On the contrary, for every $\mathbf{p}$ that fulfills \eqref{eq:83}, we can 
always construct $\{\gamma_{i,j}\}_{i\neq j\in[n]}$ with $\gamma_{i,j}\leq p_j$
such that \eqref{eq:82} holds.
Thus, Linear Programming \ref{lp:LP2} 
is equivalent to Linear Programming \ref{lp:LP6}.
\end{proof}


\section{Explicit constructions of MUB and MR-MUB codes}\label{sec:code}



\subsection{MR-MUB and MUB codes}

Based on the previous section, we can now define two particular classes
of irregular array codes.

\begin{definition}\label{def:1}
A Minimal Update Bandwidth (MUB) code 
is an $(n,k,\mathbf{m})$ irregular array code
with update bandwidth equal to $\gamma_{\min}$.
\end{definition}

\begin{definition}\label{def:2}
A Minimum Redundancy and Minimum Update Bandwidth (MR-MUB) code
is an $(n,k,\mathbf{m})$ irregular array code,
of which the code redundancy and the update bandwidth
are equal to $R_{\min}$ and $\gamma_{\min}$, respectively.
\end{definition}

Note that the existence of MR-MUB codes for certain 
parameters $n$, $k$ and $\mathbf{m}$ is not guaranteed.
In certain cases, we can only have $R_{\sma}>R_{\min}$, i.e., the smallest code redundancy subject to $\gamma=\gamma_{\min}$ is strictly larger than the minimum code redundancy among all irregular array codes.
An example is given in Fig. \ref{fig:fig3},
where we can obtain from \eqref{eq:Rsma} that
the smallest code redundancy of $(4,2,\mathbf{m}=[4 \ 2 \ 2 \ 0]^{\Tr})$ irregular array codes
is equal to
\begin{equation}
R_{\sma} = \frac{3}{2}\sum_{i=1}^{2}{m_{i}}+\frac{2}{2}\,m_{3}
=\frac{3}{2}(4+2)+2=11,
\end{equation}
while the minimum code redundancy in \eqref{eq:39} is given by
\begin{equation}
R_{\min}=([4-4]_++4)+([4-2]_{+}+2)=8.
\end{equation} 
It can be verified that the code redundancy of the irregular array code
in Fig.~\ref{fig:fig3} achieves $p_1+p_2+p_3+p_4=11=R_{\sma}$.

To confirm that the code in Fig.~\ref{fig:fig3} is an MUB code,
we note that 
the update of the first node
has to send $\Delta x_{1,1}$ and $\Delta x_{1,2}$ to node $2$, $\Delta x_{1,3}$ and $\Delta x_{1,4}$ to node $3$, $(\Delta x_{1,1}+\Delta x_{1,3})$ and $(\Delta x_{1,2}+\Delta x_{1,4})$ to node $4$, respectively. Thus, the required update bandwidth 
for node 1 is $6$. Similarly, we can verify that
the required bandwidths of the second, the third and the fourth nodes are $3$, $3$ and $0$, respectively.
As a result, $\gamma=\frac 14(6+3+3+0)=3$, which equals $\gamma_{\min}$ in 
\eqref{gamma_min}.

Two particular situations, which
guarantee the existence of MR-MUB codes, are $k=1$ and $k=n-1$.
In the former situation, we can obtain from \eqref{eq:Rsma} and \eqref{eq:39} that
\begin{equation}
R_{\sma}=R_{\min}=\frac{(n-k)}kB=\frac{(n-k)}k\sum_{i\in[n]}m_i,
\end{equation}
while the latter has been proven 
in Theorem \ref{th:xx5}. 
For $1<k<n-1$, however,
it is interesting to find that an MR-MUB code exists only when $\mathbf{m}$
is either an extremely balanced all-equal vector or an extremely unbalanced all-zero-but-one vector, which is proven in the next theorem under
$k \mid m_{i}$ for all $i \in [n]$.

\begin{theorem}\label{th:xx6}
Under $1< k<n-1$ and $k \mid m_{i}$ for all $i \in [n]$, 
$(n,k,\mathbf{m})$ MR-MUB codes exist if, and only if,
one of the two situations occurs:
\begin{align}\label{eq:67}
\begin{cases}
&m_{i}=\frac Bn\ \forall i\in[n],\\
&p_{j}=\frac{(n-k)}{nk}B\ \forall j\in[n].
\end{cases}
\end{align}
and
\begin{align}\label{eq:93pn}
\begin{cases}
&m_{1}=B, \text{ and }m_{i}=0\text{ for }2\leq i\leq n,\\
&p_{1}=0,\text{ and }p_{j}=\frac{B}{k}\text{ for }2 \leq j \leq n.
\end{cases}
\end{align}
In either situation, $\{\gamma_{i,j}\}_{i\neq j\in[n]}$
follows from \eqref{eq:sl1}.

\end{theorem}
\begin{proof} 
The theorem can be proved by simply equating the two $p_1$'s that respectively achieve   
$R_{\min}$ and $R_{\sma}$.
Specifically, \eqref{eq:47} indicates that $R_{\min}$ is achieved 
by $p_{1}=[\mu-m_{1}]_{+}$, where $\mu$
is given in \eqref{eq:x29}. From \eqref{eq:x54},
$R_{\sma}$ is reached when 
\begin{equation}
p_{1}=\frac 1k\left(\sum_{i=1}^{n-k+1}{m_{i}}-m_{1}\right)=\frac 1k\sum_{i=2}^{n-k+1}{m_{i}}.
\end{equation}
We thus have
\begin{equation}\label{eq:x65}
p_{1}=[\mu-m_{1}]_{+}=\frac 1k\sum_{i=2}^{n-k+1}{m_{i}}.
\end{equation}
We then distinguish between two cases: $p_1=0$ and $p_1>0$.

Consider $p_{1}=\frac 1k\sum_{i=2}^{n-k+1}m_i=0$, which from~\eqref{eq:mi}, immediately leads to $m_{1}=B$ and $m_{2}=m_{3}=\dots=m_{n}=0$. 
Thus, we obtain from \eqref{eq:x29} and \eqref{eq:47} that 
$\mu=\frac{B}{k}$ and $p_{j}=\frac{B}{k}$ for $2 \leq j \leq n$.
As anticipated, this $\mathbf{p}$ also satisfies \eqref{eq:x54} and validates
$R_{\min}=R_{\sma}$.

Next, we consider $p_{1}=[\mu-m_{1}]_{+} > 0$, which leads to $\mu > m_{1}$. 
As $m_{n-k} \leq m_{1}$ from~\eqref{eq:x29}, we have $\mu=\frac{B}{k} > m_{1}$. Thus, \eqref{eq:x65} becomes
\begin{equation}
\frac{B}{k}-m_{1}=\frac 1k\sum_{i=2}^{n-k+1}{m_{i}},
\end{equation}
which implies 
\begin{equation}
(k-1)m_{1}=\underbrace{m_{n-k+2}+\dots+m_{n}}_{k-1}.
\end{equation}
We can then conclude from \eqref{eq:mi} that $m_{1}=m_{2}=\dots=m_{n}$. 
The verification of $R_{\min}=R_{\sma}$ straightforwardly follows.
\end{proof}

In practice, it may be unusual to place all data symbols in one node.
Thus, we will focus on the construction of MR-MUB codes
that follows \eqref{eq:67} in the next subsection.
In other words, the $(n,k,\mathbf{m}=[m\ m\ \cdots\ m]^{\Tr})$ MR-MUB codes considered in the rest of the paper 
are $(n,k)$ vertical MDS array codes with 
each node containing $m$ data symbols and
$p=\frac{(n-k)}km$ parity symbols subject to $k\mid m$.

Note that Theorem \ref{th:xx6} seems limited in 
its applicability 
since \eqref{eq:67} simply shows 
vertical MDS codes can achieve both
$R_{\min}$ and $\gamma_{\min}$ under a particular case of $k\mid m$.
However, without the condition of $k\mid m$, vertical MDS array codes may not 
form a sub-class of MR-MUB codes.
This can be justified by two observations.
First, it can be verified from \eqref{eq:47}
that the fulfillment of both 
$p_i=[\mu-m_i]_{+}$ for $1\leq i \leq n-k$
and $\sum_{i=n-k+1}^{n}{p_i}=\sum_{i=1}^{n-k}{m_i}$
under each $p_i=p$ and each $m_i=m$
requires $k\mid nm$. Thus, under $k\nmid nm$, 
vertical MDS array codes cannot achieve the minimum code redundancy,
and hence cannot be MR-MUB codes.
Second, when $k\mid nm$ but $k\nmid m$,
examples and counterexamples for vertical MDS array codes 
being able to achieve 
simultaneously $R_{\min}$ and $\gamma_{\min}$ 
can both be constructed.\footnote{
A supporting example follows when $n=6$, $k=3$ and $m=4$, where
setting $p_i=p=4$ and 
\begin{equation}
\gamma_{i,j}=\begin{cases}
1,&j\in\{(i\,\text{mod}\,6)+1, [(i+1)\,\text{mod}\,6]+1\}\\
2,&\text{otherwise}
\end{cases}
\end{equation}
for $i\neq j\in[n]$ fulfills both \eqref{eq:x26} and \eqref{eq:x27}, and achieves 
simultaneously $R_{\min}$ and $\gamma_{\min}$.
 
A counterexample exists when $n=9$, $k=6$ and $m=2$.
From \eqref{eq:47}, we know $R_{\min}$ can only be achieved by adopting $p_j=1$ for $j\in[n]$.
By \eqref{eq:xxx''}, the achievability 
of $\gamma_{\min}$ requires  
$\gamma_{i,j_1(i)}=\gamma_{i,j_2(i)}=\gamma_{i,j_3(i)}=\gamma_{i,j_4(i)}=0$ 
and $\gamma_{i,j_5(i)}=\gamma_{i,j_6(i)}=\gamma_{i,j_7(i)}=\gamma_{i,j_8(i)}=1$ 
for every $i\in[n]$. Then, the pigeon hole principle implies that there is
$j'$ such that $\{\gamma_{i,j'}\}_{i\in[n]\setminus\{j'\}}$
contains at least four $0$'s. Let $\gamma_{i_1,j'}=\gamma_{i_2,j'}=\gamma_{i_3,j'}=\gamma_{i_4,j'}=0$ and $\mathcal{E}=\{i_1,i_2,i_3\}$, where $j'\not\in\{i_1,i_2,i_3,i_4\}$. 
A violation to \eqref{eq:x27} can thus be obtained as follows:
\begin{equation}
\begin{cases}
\sum_{i \in \mathcal{E}}{m_{i}}=|\mathcal{E}|\, m=6\\
\sum_{j \in \bar{\mathcal{E}}}{\min\big\{p_{j}, \sum_{i \in \mathcal{E}}{\gamma_{i,j}}\big\}}=\min\big\{p_{j'}, \underbrace{\mbox{$\sum_{i \in \mathcal{E}}{\gamma_{i,j'}}$}}_{=0}\big\}
+\sum_{j \in \bar{\mathcal{E}}\setminus\{j'\}}{\min\big\{p_{j}, \sum_{i \in \mathcal{E}}{\gamma_{i,j}}\big\}}\leq 5.
\end{cases}
\end{equation}
Consequently, $(9,6)$ vertical MDS array codes with each node having $m=2$ data symbols cannot be MR-MUB codes.
} Hence, we conjecture that $k\mid m$ is also a necessary condition for 
vertical MDS array codes 
being MUB codes, provided $k\nmid n$. 

Theorem~\ref{th:xx6} only deals with the situation of $1<k<n-1$. For completeness,
the next corollary incorporates also the two particular cases of $k=1$ and $k=n-1$.

\begin{corollary}
Unde $1\leq k<n$ and $k\mid m$, an $(n,k,m\one)$ MR-MUB code must parameterize with
\begin{equation}\label{eq:sl3}
p_{j}=\frac{(n-k)}km \triangleq p \quad \forall j \in [n],
\end{equation}
\begin{equation}\label{eq:sl2}
\gamma_{i,j}=\frac{m}{k} \quad \forall i \neq j \in [n],
\end{equation}
where $\one\triangleq[1\ 1\ \cdots\ 1]^{\Tr}$ is the all-one vector.
\end{corollary}
\begin{proof} 
We only substantiate the corollary for $k=1$ and $k=n-1$ 
since the situation of $1<k<n-1$ have been proved in Theorem \ref{th:xx6}.
The validity of \eqref{eq:sl3} under $k=1$ and $k=n-1$
can be confirmed by \eqref{eq:37a}. 
We can also obtain from \eqref{eq:sl1} that \eqref{eq:sl2}
holds under $k=1$. It remains to verify \eqref{eq:sl2} under $k=n-1$ by contradiction.

Fix $k=n-1$. Suppose there is a $j' \in [n]\setminus \{i\}$ such that $\gamma_{i,j'} < \frac{m}{k} = p_{j'}$. A contradiction can be established from \eqref{eq:x27} as follows:
\begin{equation}
m=m_i \leq \sum_{j \in [n]\setminus\{i\}}{\min\{p_{j}, \gamma_{i,j}\}} \leq \sum_{j \in [n] \setminus \{i,j'\}}{p_{j}}+\gamma_{i,j'}< \sum_{j \in [n]\setminus\{i\}}{p_{j}}=m.
\end{equation}
Accordingly, $\gamma_{i,j} \geq \frac{m}{k}$ for all $i \neq j \in [n]$, which implies 
\begin{equation}\label{eq:pn}
\sum_{j \in [n] \setminus \{i\}}{\gamma_{i,j}} \geq \frac mk(n-1)=m.
\end{equation}
By noting from \eqref{eq:x36'} that the inequality in \eqref{eq:pn} must 
be replaced by an equality, \eqref{eq:sl2} holds under $k=n-1$.
\end{proof}

\subsection{Construction of MR-MUB codes}\label{sec:5A}


For the construction of $(n,k,m\one)$ MR-MUB code, denoted as $\mathcal{C}_{\text{O}}$ for convenience,
we require $\mathbf{x}_{i} \in \mathbb{F}_{q}^{m}$ and $\mathbf{p}_{j} \in \mathbb{F}_{q}^{p}$ with $p=\frac{(n-k)}km$ 
for $i,j \in [n]$. The construction of $\{\mathbf{A}_{i,j}\}_{i \neq j \in [n]}$ and 
$\{\mathbf{B}_{i,j}\}_{i \neq j \in [n]}$ associated with $\mathcal{C}_{\text{O}}$
are then addressed as follows.


First, we construct $\{\mathbf{A}_{i,j}\}_{i \neq j \in [n]}$ of dimension $\frac mk\times m$. 
Choose an $(n-1,k)$ MDS array code $\mathcal{M}$ over $\mathbb{F}_{q}$ with encoding function $\mathcal{F}: \mathbb{F}_{q}^{m} \longrightarrow \mathbb{F}_{q}^{\frac{m}{k} \times (n-1)}$, where $\frac{m}{k}$ is the number of rows of the MDS array code, and $m$ is the number of data symbols 
in each row.
As an example, $\mathcal{M}$ can be a Reed-Solomon (RS) code 
subject to $q\geq n-1$.
Denote $\mathbf{F}_{i}\triangleq \mathcal{F}(\mathbf{x}_{i})$. Then, 
$\{\mathbf{p}_{i,j}\}_{i\neq j \in [n]}$ defined 
in \eqref{eq:gi}, as well as $\{\mathbf{A}_{i,j}\}_{i \neq j \in [n]}$, 
can be characterized via
\begin{equation}
\mathbf{p}_{i,[(i+j-1)\bmod n]+1}=\mathbf{A}_{i,[(i+j-1)\bmod n]+1}\mathbf{x}_{i}=(\mathbf{F}_{i})_{j}\quad\forall j\in[n],
\end{equation}
where $(\mathbf{F}_{i})_{j}$ is the $j$-th column of the matrix $\mathbf{F}_{i}$. This indicates that 
\begin{equation}\label{eq:53}
\mathbf{F}_{i}=[\mathbf{p}_{i,i+1} \ \dots \ \mathbf{p}_{i,n} \ \mathbf{p}_{i,1} \ \dots \ \mathbf{p}_{i,i-1}] \quad \forall i \in [n].
\end{equation}

Next, we construct $\{\mathbf{B}_{i,j}\}_{i \neq j \in [n]}$ of dimension $p\times \frac{m}{k}$. Choose a $p \times \frac{(n-1)m}{k} $ matrix $\mathbf{V}$ over $\mathbb{F}_{q}$ such that arbitrary selection of $p$ columns of $\mathbf{V}$ form an invertible matrix. For example,  $\mathbf{V}$ can be a Vandermonde matrix 
subject to $q\geq\frac{(n-1)}{k}m$. 
We then let 
\begin{equation}
\mathbf{B}_{[(i+j-1)\bmod n]+1,j}=[(\mathbf{V})_{(i-1)\frac{m}{k}+1} \ (\mathbf{V})_{(i-1)\frac{m}{k}+2} \dots (\mathbf{V})_{i\frac{m}{k}}]
\quad\forall i \in [n-1]\text{ and }j \in [n],
\end{equation}
which implies that
\begin{equation}\label{eq:VV}
\mathbf{V}=\begin{bmatrix}
\mathbf{B}_{j+1,j} & \dots & \mathbf{B}_{n,j} & \mathbf{B}_{1,j} & \dots & \mathbf{B}_{j-1,j}
\end{bmatrix}. 
\end{equation}
Note that the right-hand-side of \eqref{eq:VV} 
remains constant regardless of $j\in[n]$. 
Thus, we can obtain from \eqref{eq:fi} that
\begin{equation}\label{eq:pjv}
\mathbf{p}_{j}=\mathbf{V}\,\begin{bmatrix}
\mathbf{p}_{j+1,j}^{\Tr} & \dots & \mathbf{p}_{n,j}^{\Tr} & \mathbf{p}_{1,j}^{\Tr} & \dots & \mathbf{p}_{j-1,j}^{\Tr}\\  
\end{bmatrix}^{\Tr}. 
\end{equation}

We now prove the code so constructed is an MR-MUB code.

\begin{theorem}\label{th:x6}
$\mathcal{C}_{\text{O}}$ is an $(n,k,m\one)$ MR-MUB code.
\end{theorem}
%
%
%
\begin{proof} 
The proof requires verifying two properties, which are
$i)$ $\mathcal{C}_{\text{O}}$ being an $(n,k,m\one)$ 
array code, and $ii)$ $\mathcal{C}_{\text{O}}$ achieving 
$R_{\min}$ and $\gamma_{\min}$.

First, we justify $i)$, i.e., 
$\mathcal{C}_{\text{O}}$ 
satisfying that 
given any set $\mathcal{E} \subset [n]$ with $|\mathcal{E}|=n-k$, the codeword $\mathbf{C}$ of $\mathcal{C}_{\text{O}}$ can be reconstructed from $\mathbf{C}_{\bar{\mathcal{E}}}$.
When $\mathbf{C}_{\bar{\mathcal{E}}}$
is given, both $\mathbf{X}_{\bar{\mathcal{E}}}$ and $\mathbf{P}_{\bar{\mathcal{E}}}$ are known, and so are $\{\mathbf{p}_{\bar{e}_{i},\bar{e}_{j}}\}_{i\neq j\in[k]}$
according to \eqref{eq:gi}. We can then establish from \eqref{eq:fi} that
\begin{equation}\label{eq:56}
\begin{aligned}
&\mathbf{p}_{\bar{e}_{j}}-\sum_{i=1,i\neq j}^{k}{\mathbf{B}_{\bar{e}_{i},\bar{e}_{j}}\mathbf{p}_{\bar{e}_{i},\bar{e}_{j}}}
=\sum_{i=1}^{n-k}{\mathbf{B}_{e_{i},\bar{e}_{j}}\mathbf{p}_{e_{i},\bar{e}_{j}}}
=\begin{bmatrix}
\mathbf{B}_{e_{1},\bar{e}_{j}} & \dots & \mathbf{B}_{e_{n-k},\bar{e}_{j}}
\end{bmatrix}
\begin{bmatrix}
\mathbf{p}_{e_{1},\bar{e}_{j}}\\
\vdots\\
\mathbf{p}_{e_{n-k},\bar{e}_{j}}\\
\end{bmatrix} \quad \forall j \in [k],
\end{aligned}
\end{equation}
Since $\mathbf{p}_{\bar{e}_{j}}-\sum_{i=1,i\neq j}^{k}{\mathbf{B}_{\bar{e}_{i},\bar{e}_{j}}\mathbf{p}_{\bar{e}_{i},\bar{e}_{j}}}$ 
is known and any $p$ columns of $\mathbf{V}$, 
as defined in \eqref{eq:VV}, forms an invertible matrix,
we can obtain 
$\{\mathbf{p}_{e_{i},\bar{e}_{j}}\}_{i \in [n-k], j\in [k]}$ 
by left-multiplying \eqref{eq:56}
by $[\mathbf{B}_{e_{1},\bar{e}_{j}} \ \dots \ \mathbf{B}_{e_{n-k},\bar{e}_{j}}]^{-1}$.
With the knowledge of $k$ columns $\{\mathbf{p}_{e_{i},\bar{e}_{j}}\}_{j\in [k]}$ of $\mathbf{F}_{e_i}$
in \eqref{eq:53}, 
we can recover  $\mathbf{x}_{e_{i}}$ 
via the decoding algorithm of the $(n-1,k)$ MDS array code $\mathcal{M}$. 
By this procedure, $\{\mathbf{x}_{i}\}_{i \in [n]}$ can all be recovered.

Next, we verify $ii)$. From~\eqref{eq:53}, we have $\mathbf{p}_{i,j} \in \mathbb{F}_{q}^{\frac{m}{k}}$ and hence $\gamma_{i,j}=\frac{m}{k}$, 
which leads to $\gamma=\gamma_{\min}$ as pointed out in \eqref{eq:sl2}. In addition, \eqref{eq:pjv} shows
$p_{j}=\frac{(n-k)}{k}m$ for $j \in [n]$, and hence $R_{\min}$ is achieved 
as addressed in  
\eqref{eq:sl3}. The justification of the two required properties of $\mathcal{C}_{\text{O}}$ is thus completed.
\end{proof}

The $(4,2,2\,\one)$ MR-MUB code in Fig.~\ref{fig:subfig:b} can be constructed via the proposed procedure. 
 First, with $\mathbf{x}_{i}=[x_{i,1} \ x_{i,2}]^{\Tr}$, $\mathcal{M}$ is chosen as a $(3,2)$ parity-check code over $\mathbb{F}_{q}$, which gives
\begin{equation}\label{eq:74}
\mathbf{F}_{i}=\begin{bmatrix}
x_{i,1} & x_{i,2} & x_{i,1}+x_{i,2}
\end{bmatrix} \quad \forall i \in [n].
\end{equation}
Thus, from~\eqref{eq:53}, we have $\mathbf{p}_{1,2}=x_{1,1}$, $\mathbf{p}_{1,3}=x_{1,2}$, $\mathbf{p}_{1,4}=x_{1,1}+x_{1,2}$. 
The remaining $\mathbf{p}_{i,j}$ can be similarly obtained 
and are listed in Table~\ref{fig:fig4}. 

\begin{table}[b]
\caption{\label{fig:fig4} $\{\mathbf{p}_{i,j}\}_{i\neq j\in[n]}$ of the MR-MUB code presented in Fig.~\ref{fig:subfig:b},
where the element in the $i$-th row and the $j$-th column is $\mathbf{p}_{i,j}$.}
\centering
\begin{tabular}{|c|c|c|c|}
\hline
null & $x_{1,1}$ & $x_{1,2}$ & $x_{1,1}+x_{1,2}$\\
\hline
$x_{2,1}+x_{2,2}$ & null & $x_{2,1}$ & $x_{2,2}$\\
\hline
$x_{3,2}$ & $x_{3,1}+x_{3,2}$ & null & $x_{3,1}$\\
\hline
$x_{4,1}$ & $x_{4,2}$ & $x_{4,1}+x_{4,2}$ & null\\
\hline
\end{tabular}
\end{table} 

Next, we specify 
\begin{equation}
\mathbf{V}=\begin{bmatrix}
0 & 1 & 1\\
1 & 1 & 0\\
\end{bmatrix},
\end{equation}
which satisfies that the selection of any two columns 
forms an invertible matrix. By~\eqref{eq:pjv}, we have  
\begin{equation}
\mathbf{p}_{1}=\mathbf{V} \, [\mathbf{p}_{2,1}^{\Tr} \ \mathbf{p}_{3,1}^{\Tr} \ \mathbf{p}_{4,1}^{\Tr}]^{\Tr}=\begin{bmatrix}
0 & 1 & 1\\
1 & 1 & 0\\
\end{bmatrix} \, \begin{bmatrix}
x_{2,1}+x_{2,2}\\
x_{3,2}\\
x_{4,1}\\ 
\end{bmatrix}=\begin{bmatrix}
x_{3,2}+x_{4,1}\\
x_{2,1}+x_{2,2}+x_{3,2}
\end{bmatrix}.
\end{equation}
$\mathbf{p}_{2}$, $\mathbf{p}_{3}$ and $\mathbf{p}_{4}$ 
can be similarly obtained and can be found in Fig.~\ref{fig:subfig:b}. 


We now demonstrate via this example how erased nodes can be systematically recovered based on the chosen $\mathcal{M}$ and $\mathbf{V}$.
Suppose nodes~$1$ and~$2$ are erased.
As knowing from \eqref{eq:pjv} that
\begin{equation}
\mathbf{p}_{3}=\begin{bmatrix}
0 &1 &1 \\
1 & 1 & 0\\
\end{bmatrix} \, \begin{bmatrix}
\mathbf{p}_{4,3}\\
\mathbf{p}_{1,3}\\
\mathbf{p}_{2,3}\\
\end{bmatrix},
\end{equation} 
we perform \eqref{eq:56} to obtain
\begin{equation}\label{left}
\mathbf{p}_{3}-\begin{bmatrix}
0\\
1\\
\end{bmatrix}\mathbf{p}_{4,3}=\begin{bmatrix}
1 & 1\\
1 & 0\\
\end{bmatrix}\begin{bmatrix}
\mathbf{p}_{1,3}\\
\mathbf{p}_{2,3}\\
\end{bmatrix}.
\end{equation}
Since $\mathbf{p}_{3}$ is known and 
$\mathbf{p}_{4,3}$ can be obtained from $\mathbf{x}_{4}$ via $\mathbf{p}_{4,3}=\mathbf{A}_{4,3}\mathbf{x}_{4}$, 
we can recover $\mathbf{p}_{1,3}$ and $\mathbf{p}_{2,3}$ via 
\begin{equation}\label{just1}
\begin{bmatrix}
\mathbf{p}_{1,3}\\
\mathbf{p}_{2,3}\\
\end{bmatrix}=\begin{bmatrix}
1 & 1\\
1 & 0\\
\end{bmatrix}^{-1}\left(\mathbf{p}_{3}-\begin{bmatrix}
0\\
1\\
\end{bmatrix}\mathbf{p}_{4,3}\right).
\end{equation}
The recovery of $\mathbf{p}_{1,4}$ and $\mathbf{p}_{2,4}$ can be similarly done
via
\begin{equation}\label{just2}
\begin{bmatrix}
\mathbf{p}_{1,4}\\
\mathbf{p}_{2,4}\\
\end{bmatrix}=\begin{bmatrix}
0 & 1\\
1 & 1\\
\end{bmatrix}^{-1}\left(\mathbf{p}_{4}-\begin{bmatrix}
1\\
0\\
\end{bmatrix}\mathbf{p}_{3,4}\right).
\end{equation}
We then note from \eqref{eq:53} that $\mathbf{F}_1=\mathcal{F}(\mathbf{x}_1)=[\mathbf{p}_{1,2} \ \mathbf{p}_{1,3} \ \mathbf{p}_{1,4}]$ is a codeword  
of $\mathcal{M}$, corresponding to $\mathbf{x}_{1}$, and
its second and third columns are just recovered via \eqref{just1} and \eqref{just2}.
By equating the second and the third columns of $\mathbf{F}_1$ 
with \eqref{eq:74},  
the recovery of $\mathbf{x}_{1}$ is done.
We can similarly recover $\mathbf{x}_2$ by using
the recovered $\mathbf{p}_{2,3}$ and $\mathbf{p}_{2,4}$ in \eqref{just1} and \eqref{just2}. 
The recovery of the two erased nodes is thus completed.

\subsection{Construction of MUB codes with the smallest code redundancy}\label{sec:6b}

We  continue to propose a construction of $(n,k,\mathbf{m})$ MUB codes with the smallest code redundancy, and denote the code to be constructed as $\mathcal{C}_{\text{U}}$ for notational convenience.
This can be considered a generalization of the code construction in the previous subsection.

For the construction of $\mathcal{C}_{\text{U}}$,
we require $\mathbf{x}_{i} \in \mathbb{F}_{q}^{m_i}$ and $\mathbf{p}_{j} \in \mathbb{F}_{q}^{p_j}$ with $\{p_j\}_{j\in[n]}$ specified in \eqref{eq:x54}
for $i,j \in [n]$. The construction of $\{\mathbf{A}_{i,j}\}_{i \neq j \in [n]}$ and 
$\{\mathbf{B}_{i,j}\}_{i \neq j \in [n]}$ associated with $\mathcal{C}_{\text{U}}$
are then addressed as follows.

First, for $i\neq j\in[n]$, we construct $\mathbf{A}_{i,j}$ of dimension $\frac {m_i}k\times m_i$. For each $i\in[n]$, choose an $(n-1,k)$ MDS array code $\mathcal{M}_i$ over $\mathbb{F}_{q}$ with encoding function $\mathcal{F}_i: \mathbb{F}_{q}^{m_i} \longrightarrow \mathbb{F}_{q}^{\frac{m_i}{k} \times (n-1)}$, where $\frac{m_i}{k}$ is the number of rows of the MDS array code, and $m_i$ is the number of data symbols 
in each row.
Denote $\mathbf{F}_{i}\triangleq \mathcal{F}_i(\mathbf{x}_{i})$. 
Then, 
$\{\mathbf{p}_{i,j}\}_{i\neq j \in [n]}$ defined 
in \eqref{eq:gi}, as well as $\{\mathbf{A}_{i,j}\}_{i \neq j \in [n]}$, 
can be characterized via
\begin{equation}
\mathbf{p}_{i,[(i+j-1)\bmod n]+1}=\mathbf{A}_{i,[(i+j-1)\bmod n]+1}\mathbf{x}_{i}=(\mathbf{F}_{i})_{j}\quad\forall j\in[n].
\end{equation}
This indicates that 
\begin{equation}\label{eq:x83}
\mathbf{F}_{i}=[\mathbf{p}_{i,i+1} \ \dots \ \mathbf{p}_{i,n} \ \mathbf{p}_{i,1} \ \dots \ \mathbf{p}_{i,i-1}] \quad \forall i \in [n].
\end{equation}

Next, for $i\neq j\in[n]$, we construct $\mathbf{B}_{i,j}$ of dimension $p_j\times \frac{m_i}{k}$. Choose a $p_j \times \sum_{i\in[n]\setminus\{j\}}\frac{m_i}{k} $ matrix $\mathbf{V}_j$ over $\mathbb{F}_{q}$ such that arbitrary selection of $p_j$ columns of $\mathbf{V}_j$ form an invertible matrix. 
We then get $\{\mathbf{B}_{i,j}\}_{i\neq j\in[n]}$ from 
\begin{equation}\label{eq:VVV}
\mathbf{V}_{j}=\begin{bmatrix}
\mathbf{B}_{j+1,j} & \dots & \mathbf{B}_{n,j} & \mathbf{B}_{1,j} & \dots & \mathbf{B}_{j-1,j}\\  
\end{bmatrix} \quad \forall j\in[n].
\end{equation}
Thus, we can obtain from \eqref{eq:fi} that
\begin{equation}\label{eq:pjv'}
\mathbf{p}_{j}=\mathbf{V}_{j}\,\begin{bmatrix}
\mathbf{p}_{j+1,j}^{\Tr} & \dots & \mathbf{p}_{n,j}^{\Tr} & \mathbf{p}_{1,j}^{\Tr} & \dots & \mathbf{p}_{j-1,j}^{\Tr}\\  
\end{bmatrix}^{\Tr}. 
\end{equation}

We now prove the code so constructed is an MUB code with the smallest code redundancy.

\begin{theorem}
$\mathcal{C}_{\text{U}}$ is an $(n,k,\mathbf{m})$ MUB code with the smallest code redundancy.
\end{theorem}
\begin{proof} 
Similar to the proof of Theorem~\ref{th:x6}, 
the substantiation of this theorem requires verifying two properties:
$i)$ $\mathcal{C}_{\text{U}}$ is an $(n,k,\textbf{m})$ 
array code, and $ii)$ $\mathcal{C}_{\text{U}}$ achieves $R_{\sma}$ and 
$\gamma_{\min}$.

First, we justify $i)$, i.e., $\mathcal{C}_{\text{U}}$ 
satisfying that 
given any set $\mathcal{E} \subset [n]$ with $|\mathcal{E}|=n-k$, the codeword $\mathbf{C}$ of $\mathcal{C}_{\text{U}}$ can be reconstructed from $\mathbf{C}_{\bar{\mathcal{E}}}$.
When $\mathbf{C}_{\bar{\mathcal{E}}}$
is given, both $\mathbf{X}_{\bar{\mathcal{E}}}$ and $\mathbf{P}_{\bar{\mathcal{E}}}$ are known, and so are $\{\mathbf{p}_{\bar{e}_{i},\bar{e}_{j}}\}_{i\neq j\in[k]}$
according to \eqref{eq:gi}. We can then establish from \eqref{eq:fi} that
\begin{equation}\label{eq:x86}
\begin{aligned}
&\mathbf{p}_{\bar{e}_{j}}-\sum_{i=1,i\neq j}^{k}{\mathbf{B}_{\bar{e}_{i},\bar{e}_{j}}\mathbf{p}_{\bar{e}_{i},\bar{e}_{j}}}
=\begin{bmatrix}
\mathbf{B}_{e_{1},\bar{e}_{j}} & \dots & \mathbf{B}_{e_{n-k},\bar{e}_{j}}
\end{bmatrix}
\begin{bmatrix}
\mathbf{p}_{e_{1},\bar{e}_{j}}\\
\vdots\\
\mathbf{p}_{e_{n-k},\bar{e}_{j}}\\
\end{bmatrix} \quad \forall j \in [k].
\end{aligned}
\end{equation}
According to \eqref{eq:x54}, we have 
\begin{equation}\label{eq:118pn}
\row([\mathbf{B}_{e_{1},\bar{e}_{j}} \ \dots \ \mathbf{B}_{e_{n-k},\bar{e}_{j}}])=p_{j} \geq \sum_{i \in [n-k]}{\frac{m_{e_{i}}}{k}} =\sum_{i \in [n-k]}{\col(\mathbf{B}_{e_{i},\bar{e}_{j}})}.
\end{equation}
Since 
any $p_j$ columns of $\mathbf{V}_j$, 
as defined in \eqref{eq:VVV}, forms an invertible matrix,
we obtain from \eqref{eq:118pn} that $[\mathbf{B}_{e_{1},\bar{e}_{j}}$ $\dots$ $ \mathbf{B}_{e_{n-k},\bar{e}_{j}}]$ is of full column rank, and 
hence $\{\mathbf{p}_{e_{i},\bar{e}_{j}}\}_{i \in [n-k], j\in [k]}$ 
can be solved via~\eqref{eq:x86}.
With the knowledge of $k$ columns $\{\mathbf{p}_{e_{i},\bar{e}_{j}}\}_{j\in [k]}$ of $\mathbf{F}_{e_i}$
in \eqref{eq:x83}, 
we can recover  $\mathbf{x}_{e_{i}}$ 
via the decoding algorithm of the $(n-1,k)$ MDS array code $\mathcal{M}_{e_i}$. 
By this procedure, $\{\mathbf{x}_{i}\}_{i \in [n]}$ can all be recovered.

Next, we verify $ii)$. From~\eqref{eq:x83}, we have $\mathbf{p}_{i,j} \in \mathbb{F}_{q}^{\frac{m_i}{k}}$ and hence $\gamma_{i,j}=\frac{m_i}{k}$, 
which leads to $\gamma=\gamma_{\min}$ as pointed out in \eqref{eq:sl1}. In addition, \eqref{eq:pjv'} shows
$\{p_{j}\}_{j\in[n]}$ follows \eqref{eq:x54}, and hence $R_{\sma}$ is achieved 
as addressed in  Theorem \ref{th:x4}.
The justification of the two required properties of $\mathcal{C}_{\text{U}}$ is thus completed.
\end{proof}

We demonstrate that the $(4,2,\mathbf{m}=[4 \ 2 \ 2 \ 0]^{\Tr})$ MUB code in Fig.~\ref{fig:fig3} 
can be constructed via the proposed procedure.
First, with $\mathbf{x}_{1}=[x_{1,1} \dots x_{1,4}]^{\Tr}$, 
$\mathcal{M}_{1}$ is chosen as a $(3,2)$ MDS array code,
which encodes $\mathbf{x}_1$ into
\begin{equation}
\mathbf{F}_{1}=\begin{bmatrix}
x_{1,1} & x_{1,3} & x_{1,1}+x_{1,3}\\
x_{1,2} & x_{1,4} & x_{1,2}+x_{1,4}\\
\end{bmatrix}=\begin{bmatrix}\mathbf{p}_{1,2} & \mathbf{p}_{1,3} &\mathbf{p}_{1,4}
\end{bmatrix}.
\end{equation} 
For $i=2$ and $3$,  
$\mathcal{M}_{i}$ is chosen 
to be a $(3,2)$ parity check code over $\mathbb{F}_{q}$, as 
the one in \eqref{eq:74}.
Since $m_4=0$, $\{\mathbf{p}_{4,j}\}_{j\in[4]\setminus\{4\}}$
are null vectors.
The resulting $\{\mathbf{p}_{i,j}\}_{i\neq j\in[n]}$ 
are listed in Table~\ref{fig:fig5}. 

\begin{table}[b]
\caption{\label{fig:fig5} $\{\mathbf{p}_{i,j}\}_{i\neq j\in[n]}$ of the MUB code presented in Fig.~\ref{fig:fig3}, where the element in the $i$-th row and the $j$-th column is $\mathbf{p}_{i,j}$.}
\centering
\begin{tabular}{|c|c|c|c|}
\hline
null & $\Bigg[\begin{array}{c}
x_{1,1}\\
x_{1,2}\\
\end{array}\Bigg]$
& $\Bigg[\begin{array}{c}
x_{1,3}\\
x_{1,4}\\
\end{array}\Bigg]$
& $\Bigg[\begin{array}{c}
x_{1,1}+x_{1,3}\\
x_{1,2}+x_{1,4}\\
\end{array}\bigg]$\\[3mm]\hline
$x_{2,1}+x_{2,2}$ & null & $x_{2,1}$ & $x_{2,2}$ \\
\hline
$x_{3,2}$ & $x_{3,1}+x_{3,2}$ & null & $x_{3,1}$ \\
\hline
null&null&null&null\\\hline
\end{tabular}
\end{table} 

Next, we obtain from~\eqref{eq:x54} that
$p_{1}=2$ and $p_{2}=p_{3}=p_{4}=3$, and specify
\begin{equation}\label{eq:v4}
\mathbf{V}_{1}=\begin{bmatrix}
1 & 0\\
0 & 1\\
\end{bmatrix}, \ \mathbf{V}_{2}=\mathbf{V}_{3}=\begin{bmatrix}
1 & 0 & 0\\
0 & 1 & 0\\
0 & 0 & 1\\
\end{bmatrix},
\quad\text{and}\quad
\mathbf{V}_{4}=\begin{bmatrix}
1 & 0 &0 &1 \\
0 & 1 & 0 & 1\\
0 & 0 & 1 & 1\\
\end{bmatrix},
\end{equation}
where the selection of any $p_i$ columns from $\mathbf{V}_i$ forms an invertible matrix.
By~\eqref{eq:pjv'} and Table~\ref{fig:fig5}, we have 
\begin{equation}\label{eq:p14}
\begin{aligned}
\mathbf{p}_{1}=\begin{bmatrix}
x_{2,1}+x_{2,2}\\
x_{3,2}\\
\end{bmatrix}, \quad
\mathbf{p}_{2}=\begin{bmatrix}
x_{1,1}\\
x_{1,2}\\
x_{3,1}+x_{3,2}\\
\end{bmatrix}, \quad
\mathbf{p}_{3}=\begin{bmatrix}
x_{1,3}\\
x_{1,4}\\
x_{2,1}\\
\end{bmatrix},\quad\text{and}\quad
\mathbf{p}_{4}
=\begin{bmatrix}
x_{1,1}+x_{1,3}+x_{3,1}\\
x_{1,2}+x_{1,4}+x_{3,1}\\
x_{2,2}+x_{3,1}\\
\end{bmatrix},
\end{aligned}
\end{equation}
as presented in Fig.~\ref{fig:fig3}. 

Based on this example, the systematic recovery of erased nodes 
can be demonstrated as follows.
Suppose nodes~$1$ and~$2$ are erased. Then, through 
\eqref{eq:pjv'},~\eqref{eq:x86} and~\eqref{eq:v4}, we have 
\begin{equation}
\mathbf{p}_{3}=\begin{bmatrix}
\mathbf{p}_{1,3}\\
\mathbf{p}_{2,3}\\
\end{bmatrix}, \quad\text{and}\quad \mathbf{p}_{4}-\begin{bmatrix}
1\\
1\\
1\\
\end{bmatrix} \, \mathbf{p}_{3,4}=\begin{bmatrix}
1 & 0 & 0\\
0 & 1 & 0\\
0 & 0 & 1\\
\end{bmatrix} \, \begin{bmatrix}
\mathbf{p}_{1,4}\\
\mathbf{p}_{2,4}\\
\end{bmatrix}.
\end{equation}
We can thus obtain $\mathbf{p}_{1,3}$, $\mathbf{p}_{1,4}$, $\mathbf{p}_{2,3}$ and $\mathbf{p}_{2,4}$. 
By noting $\mathbf{p}_{1,3}=[x_{1,3} \quad x_{1,4}]^{\Tr}$ and $\mathbf{p}_{1,4}=[x_{1,1}+x_{1,3} \quad x_{1,2}+x_{1,4}]^{\Tr}$,
the recovery of $\mathbf{x}_{1}$ is done
via the erasure correcting of $\mathcal{M}_{1}$.
We can similarly recover $\mathbf{x}_{2}$ from $\mathbf{p}_{2,3}$ and $\mathbf{p}_{2,4}$. 
The recovery of the two erased nodes is therefore completed.

\section{Update complexity of MR-MUB codes}\label{sec:UC}

The update complexity of an array code, denoted as $\theta$, is defined as the average number of parity symbols affected by updating a single data symbol \cite{blaum1996}.
For an $(n,k,\mathbf{m})$ irregular array codes, a definition-implied lower bound for update complexity is $\theta \geq n-k$. This lower bound can be easily justified by contradiction.
If $\theta<n-k$, then at most $(n-k-1)$ nodes are affected when updating a data symbol, which leads to a contradiction that this data symbol cannot be reconstructed by the remaining $k$ unaffected nodes.

Previous results on update complexity indicate that
the lower bound $n-k$ is not attainable by $(n,k)$ horizontal MDS array codes with $1<k<n-1$ \cite{blaum1996}. 
Later, Xu and Bruck \cite{xu1999:b} introduced an $(n,k)$ vertical MDS array code 
that can achieve $\theta=n-k$. Since 
the proposed $(n,k,m\,\one)$ MR-MUB codes in the previous section are a class of vertical MDS array codes, a query that naturally follows
is whether or not the update complexity of MR-MUB codes can reach the definition-implied lower bound. Unfortunately, 
we found the answer is negative under $k > 1$,
and will show in Theorem~\ref{th:x5} that the update complexity of MR-MUB codes is lower-bounded by $n-k+\frac{k-1}k$.



In order to facilitate the presentation of the result in Theorem \ref{th:x5},
five lemmas are addressed first.
The first lemma indicates it suffices to consider the MR-MUB codes 
with $\{\mathbf{M}_{i,i}\}_{i\in[n]}$ being zero matrices; hence, we do not need to consider $\{\mathbf{M}_{i,i}\}_{i\in[n]}$ in the calculation of update complexity (cf.~Lemma \ref{le:cii}). The second lemma shows that for the determination 
of a lower bound of update complexity, we can focus
on the decomposition of 
$\mathbf{M}_{i,j}=\mathbf{B}_{i,j}'\mathbf{A}_{i,j}'$ with $\mathbf{A}_{i,j}'$
containing an $\gamma_{i,j}\times\gamma_{i,j}$ identity submatrix.
As a result, $\mathbf{B}_{i,j}'$ is a submatrix of $\mathbf{M}_{i,j}$ 
and the column weights of $\mathbf{M}_{i,j}$ are lower-bounded 
by the column weights of $\mathbf{B}_{i,j}'$ 
(cf.~Lemma~\ref{le:cij}). The next two lemmas
then study the column weights of general  $\mathbf{B}_{i,j}$
that is not necessarily a submatrix of $\mathbf{M}_{i,j}$ (cf.~Lemmas~\ref{le:invert} and \ref{le:bij}).
The last lemma accounts for 
the number of non-zero columns in $\mathbf{M}_i^{(\ell)}\triangleq[(\mathbf{M}_{i,1})_{\ell} \ \dots (\mathbf{M}_{i,i-1})_{\ell} \quad (\mathbf{M}_{i,i+1})_{\ell} \dots \ (\mathbf{M}_{i,n})_{\ell}]$, where $(\mathbf{M}_{i,j})_{\ell}$ denotes the $\ell$-th column of matrix $\mathbf{M}_{i,j}$,


\begin{lemma}\label{le:cii}
For any $(n,k,\mathbf{m})$ irregular array code $\mathcal{C}$
with construction matrices $\{\mathbf{M}_{i,j}\}_{i,j\in[n]}$,
we can construct another $(n,k,\mathbf{m})$ irregular array code $\mathcal{C}'$
with $\{\mathbf{M}'_{i,i}=[\mathbf{0}]\}_{i\in[n]}$
such that both codes have the same code redundancy and update bandwidth.
\end{lemma}
\begin{proof}
Let the construction matrices of $\mathcal{C}'$ be defined as
\begin{equation}\label{eq:cij1}
\mathbf{M}_{i,j}'=\begin{cases}
\mathbf{M}_{i,j} & i\neq j;\\
[\mathbf{0}] & i=j.
\end{cases}
\end{equation}
Then, there exists an invertible mapping  
between codewords of 
$\mathcal{C}'$ and $\mathcal{C}$, i.e., 
\begin{equation}\label{eq:con1}
\mathbf{c}_{j}'=\begin{bmatrix}
\mathbf{x}_{j}\\
\mathbf{p}_{j}'
\end{bmatrix}
=\begin{bmatrix}
\mathbf{x}_{j}\\
\mathbf{p}_{j}-\mathbf{M}_{j,j}\mathbf{x}_j
\end{bmatrix}=
\begin{bmatrix}
\mathbf{I}&[\mathbf{0}]\\
-\mathbf{M}_{j,j}&\mathbf{I}
\end{bmatrix}
\begin{bmatrix}
\mathbf{x}_j\\
\mathbf{p}_j
\end{bmatrix}
=
\begin{bmatrix}
\mathbf{I}&[\mathbf{0}]\\
-\mathbf{M}_{j,j}&\mathbf{I}
\end{bmatrix}
\mathbf{c}_j
\quad\text{for } j \in [n],
\end{equation}
where $\mathbf{I}$ denotes an identity matrix of proper size.
A consequence of \eqref{eq:con1} is that all data symbols can be retrieved 
by accessing any $k$ columns of the corresponding codeword of $\mathcal{C}$ if, and only if,
the same can be done by accessing any $k$ columns of the corresponding codeword of $\mathcal{C}'$.
As $\mathcal{C}$ is an $(n,k,\mathbf{m})$ irregular array code,
we confirm that $\mathcal{C}'$ is also an $(n,k,\mathbf{m})$ irregular array code.
Since $\gamma_{i,j}'=\rank(\mathbf{M}_{i,j}')=\rank(\mathbf{M}_{i,j})=\gamma_{i,j}$ with $i\neq j\in[n]$,
the update bandwidth of $\mathcal{C}'$ 
remains the same as that of $\mathcal{C}$ according to \eqref{eq:defgamma}.
The relation of $\mathbf{p}_j'=\mathbf{p}_j-\mathbf{M}_{j,j}\mathbf{x}_j$
indicates $p_j'=\row(\mathbf{p}_j')=\row(\mathbf{p}_j)=p_j$ for $j\in[n]$, 
confirming $\mathcal{C}'$ and $\mathcal{C}$ have the same code redundancy.
The lemma is therefore substantiated.
\end{proof}

For an $(n,k,\mathbf{m})$ irregular array codes, 
the number of symbols affected by the update of the $\ell$-th symbol in $\mathbf{x}_{i}$ is 
\begin{equation}
\theta_{i}^{(\ell)}=\sum_{j \in [n]\setminus i}{\wt((\mathbf{M}_{i,j})_{\ell})},
\end{equation}
 and we can now omit $\mathbf{M}_{i,i}$ due to Lemma \ref{le:cii}. 
The update complexity $\theta$ of an irregular array code  
is therefore given by
\begin{equation}\label{eq:the}
\theta=\frac 1B \sum_{i\in[n]}\sum_{\ell \in [m_{i}]}{\theta_{i}^{(\ell)}}.
\end{equation}
Since the update complexity is only related to the column weights of construction matrices, the next lemma provides a structure to be considered in the calculation of $\theta$ in \eqref{eq:the}.

\begin{lemma}\label{le:cij}
There exists a full rank decomposition 
of construction matrix $\mathbf{M}_{i,j}=\mathbf{B}_{i,j}'\mathbf{A}_{i,j}'$ such that $\mathbf{A}_{i,j}'$ contains a $\gamma_{i,j} \times \gamma_{i,j}$ identity submatrix.
\end{lemma}
\begin{proof}
The existence of a full rank decomposition $\mathbf{M}_{i,j}=\mathbf{B}_{i,j}\mathbf{A}_{i,j}$ has been confirmed in Section \ref{sec:ub}.
As $\mathbf{A}_{i,j}$ is with full row rank, there exists an invertible matrix $\mathbf{R}_{i,j}$ such that $\mathbf{A}_{i,j}'=\mathbf{R}_{i,j}\mathbf{A}_{i,j}$, where $\mathbf{A}_{i,j}'$ contains a $\gamma_{i,j} \times \gamma_{i,j}$ identity submatrix. We can then obtain a new full rank decomposition $\mathbf{M}_{i,j}=\mathbf{B}_{i,j}'\mathbf{A}_{i,j}'$ with $\mathbf{B}_{i,j}'=\mathbf{B}_{i,j}\mathbf{R}_{i,j}^{-1}$. 
\end{proof}

When $\mathbf{A}_{i,j}'$ contains a $\gamma_{i,j} \times \gamma_{i,j}$ identity submatrix, 
$\mathbf{B}_{i,j}'$ must be a submatrix of $\mathbf{M}_{i,j}$. 
Thus, the column weights of $\mathbf{M}_{i,j}$ are lower-bounded by the column weights of $\mathbf{B}_{i,j}'$. 

This brings up the study 
of the next two lemmas, which hold not just for a submatrix $\mathbf{B}_{i,j}'$
of $\mathbf{M}_{i,j}$ but for general full-rank decomposition $\mathbf{B}_{i,j}$.

\begin{lemma}\label{le:invert}
Given an $(n,k,m\one)$ MR-MUB code with construction matrices  
$\{\mathbf{M}_{i,j}=\mathbf{B}_{i,j}\mathbf{A}_{i,j}\}_{i\neq j\in[n]}$,
$\mathbf{B}_{\mathcal{E},j}\triangleq [\mathbf{B}_{e_{1}, j}\ \dots\ \mathbf{B}_{e_{n-k}, j}]$
is an invertible matrix for every $\mathcal{E}\subset[n]$
with $|\mathcal{E}|=n-k$ and for every $j \notin \mathcal{E}$.
\end{lemma}
\begin{proof}
First, we note respectively from \eqref{eq:sl2} and \eqref{eq:sl3} that
$\gamma_{i,j}=\frac mk$ and $p_j=\frac{(n-k)m}k$. Hence,
$\mathbf{B}_{e_i,j}$ is an $\frac{(n-k)m}k\times\frac mk$ matrix,
implying $\mathbf{B}_{\mathcal{E},j}$ is an $\frac{(n-k)m}k\times\frac{(n-k)m}k$
square matrix. We then prove the lemma by contradiction.

Suppoe $\mathbf{B}_{\mathcal{E},j}$ is not invertible for some $\mathcal{E}$ 
with $|\mathcal{E}|=n-k$ and some 
$j\not\in\mathcal{E}$. Then, $\rank(\mathbf{B}_{\mathcal{E},j})
<\frac{(n-k)m}k$.
According to  \eqref{eq:gi} and \eqref{eq:fi}, we have
\begin{IEEEeqnarray}{rCl}\label{eq:48}
\mathbf{p}_{j}&=&\sum_{\ell \in \bar{\mathcal{E}}}{\mathbf{B}_{\ell,j}\mathbf{p}_{\ell, j}}+\sum_{i \in [n-k]}{\mathbf{B}_{e_{i},j}\mathbf{p}_{e_{i},j}}\\
&=&\sum_{\ell\in \bar{\mathcal{E}}}\mathbf{B}_{\ell,j}\mathbf{A}_{\ell,j}\mathbf{x}_\ell+\sum_{i \in [n-k]}{\mathbf{B}_{e_{i},j}\mathbf{p}_{e_{i},j}}\\
&=&[\mathbf{B}_{\bar{e}_1,j}\mathbf{A}_{\bar{e}_1,j}\ \cdots\ \mathbf{B}_{\bar{e}_k,j}\mathbf{A}_{\bar{e}_k,j}]\mathbf{X}_{\bar{\mathcal{E}}}+\mathbf{B}_{\mathcal{E},j}\mathbf{p}_{\mathcal{E},j},
\end{IEEEeqnarray}
where $\mathbf{p}_{\mathcal{E},j}\triangleq[\mathbf{p}_{e_{1}, j}^{\Tr} \ \dots \ \mathbf{p}_{e_{n-k}, j}^{\Tr}]^{\Tr}$. This implies
\begin{equation}
H(\mathbf{p}_{j} \mid \mathbf{X}_{\bar{\mathcal{E}}})=H(\mathbf{B}_{\mathcal{E},j}\mathbf{p}_{\mathcal{E},j} \mid \mathbf{X}_{\bar{\mathcal{E}}})\leq H(\mathbf{B}_{\mathcal{E},j}\mathbf{p}_{\mathcal{E},j})\leq\rank(\mathbf{B}_{\mathcal{E},j}),
\end{equation}
where the last inequality follows from Lemma~\ref{le:x1}.
We then derive based on \eqref{eq:ice} that
\begin{IEEEeqnarray}{rCl}
I(\mathbf{C}_{\bar{\mathcal{E}}} ; \mathbf{X}_{\mathcal{E}})=I(\mathbf{X}_{\mathcal{E}}; \mathbf{P}_{\bar{\mathcal{E}}} \mid \mathbf{X}_{\bar{\mathcal{E}}})&=&H(\mathbf{P}_{\bar{\mathcal{E}}} \mid \mathbf{X}_{\bar{\mathcal{E}}})\\
&\leq &\sum_{\ell \in  \bar{\mathcal{E}} \setminus \{j\}}{H(\mathbf{p}_{\ell} \mid \mathbf{X}_{\bar{\mathcal{E}}})} + H(\mathbf{p}_{j} \mid \mathbf{X}_{\bar{\mathcal{E}}})\\
&\leq & \sum_{\ell \in  \bar{\mathcal{E}} \setminus\{j\}}{H(\mathbf{p}_{\ell})}  + H(\mathbf{p}_{j} \mid \mathbf{X}_{\bar{\mathcal{E}}})\\
&\leq & \sum_{\ell \in  \bar{\mathcal{E}} \setminus\{j\}}{p_\ell}  + 
\rank(\mathbf{B}_{\mathcal{E},j})\\
&<&(n-k)m=H(\mathbf{X}_{\mathcal{E}}),\label{eq:pn135}
\end{IEEEeqnarray}
where the last strict inequality holds due to $\rank(\mathbf{B}_{\mathcal{E},j})<\frac{(n-k)m}k$.
The derivation in \eqref{eq:pn135}
indicates that $\mathbf{X}_{\mathcal{E}}$ cannot be reconstructed
from $\mathcal{C}_{\bar{\mathcal{E}}}$, 
leading to a contradiction to the definition of $(n,k,m\one)$ array codes.
\end{proof}

\begin{lemma}\label{le:bij}
For an $(n,k,m\one)$ MR-MUB code with construction matrices  
$\{\mathbf{M}_{i,j}$ $=$ $\mathbf{B}_{i,j}\mathbf{A}_{i,j}\}_{i\neq j\in[n]}$,  
\begin{equation}\label{eq:pn136}
\mathbf{B}_{j}\triangleq [\mathbf{B}_{1,j} \ \dots \mathbf{B}_{j-1,j} \ \mathbf{B}_{j+1,j} \ \dots \mathbf{B}_{n,j}]
\end{equation}
must contain at least $\frac{(k-1)m}{k}$ columns whose weight is no less than $2$.
\end{lemma}
\begin{proof}
Lemma \ref{le:invert}
shows $\mathbf{B}_{\mathcal{E},j}$ is invertible
for arbitrary $\mathcal{E}$ with $|\mathcal{E}|=n-k$; hence,
$\mathbf{B}_{j}$ in \eqref{eq:pn136} contains no zero column, 
and also has no identical columns.
As each column of $\mathbf{B}_{j}$ consists of $\frac{(n-k)m}{k}$ components,
the number of weight-one columns of $\mathbf{B}_{j}$
must be at most $\frac{(n-k)m}{k}$.
We thus conclude that there are at least
\begin{equation}
\col(\mathbf{B}_{j})-\frac{(n-k)m}{k}=\frac{(n-1)m}{k}-\frac{(n-k)m}{k}=\frac{(k-1)m}{k}
\end{equation} columns of $\mathbf{B}_j$ with weights no less than $2$. This completes the proof.
\end{proof}

As previously mentioned, since the above two lemmas hold for the 
full-rank submatrix $\mathbf{B}_{i,j}'$ of $\mathbf{M}_{i,j}$,
a lower bound of update complexity can thus be established.

\begin{corollary}\label{co:the2}
The construction matrices $\{\mathbf{M}_{i,j}\}_{i\neq j \in [n]}$ of an $(n,k,m\one)$ MR-MUB code
must have at least $\frac{(k-1)mn}{k}$ columns with weights 
no less than $2$.
\end{corollary}
\begin{proof}
Lemma \ref{le:bij} holds for 
those full-rank submatrices $\{\mathbf{B}_{i,j}'\}_{i\neq j\in[n]}$
of $\{\mathbf{M}_{i,j}\}_{i\neq j\in[n]}$. Thus,
there are at least $\frac{(k-1)m}{k}$ columns in $[\mathbf{M}_{1,j} \ \dots \mathbf{M}_{j-1,j} \ \mathbf{M}_{j+1,j} \ \dots \mathbf{M}_{n,j}]$, which 
have weights larger than $1$.
Consequently, the number of columns with weights no less than $2$
in $\{\mathbf{M}_{i,j}\}_{i\neq j \in [n]}$ is at least $\frac{(k-1)mn}{k}$.
\end{proof}


\begin{lemma}\label{le:cij1}
Fix an $(n,k,m\one)$ MR-MUB code.
For every $i\in [n]$ and every $\ell \in [m]$, there are at least $n-k$ columns with non-zero weights  in $\mathbf{M}_i^{(\ell)}\triangleq[(\mathbf{M}_{i,1})_{\ell} \ \dots (\mathbf{M}_{i,i-1})_{\ell} \quad (\mathbf{M}_{i,i+1})_{\ell} \dots \ (\mathbf{M}_{i,n})_{\ell}]$.
\end{lemma}
\begin{proof}
Denote the data symbol in the $\ell$-th row of $\mathbf{x}_{i}$ as $x_{i,\ell}$. 
If there were $k$ zero columns in $\mathbf{M}_i^{(\ell)}$,
then we can find $k$ parity vectors that are functionally independent 
of $x_{i,\ell}$ according to \eqref{eq:constr}, which implies
we can find $k$ nodes that cannot be used 
to reconstruct $x_{i,\ell}$.
A contradiction to the definition of $(n,k,m\one)$ MR MUB codes
is obtained. 
\end{proof}

Considering Lemma \ref{le:cij1} holds for every $i\in[n]$ and $\ell\in[m]$,
an immediate consequence is summarized in the next corollary.

\begin{corollary}\label{co:the1}
For an $(n,k,m\one)$ MR-MUB code, there are at least $(n-k)mn$ columns with nonzero weights  in all construction matrices $\{\mathbf{M}_{i,j}\}_{i\neq j \in [n]}$.
\end{corollary}

Corollaries \ref{co:the2} and \ref{co:the1} then lead to the main result in this section.

\begin{theorem}\label{th:x5}
The update complexity $\theta$ of $(n,k,m\one)$ MR-MUB codes is 
lower-bounded by $n-k+\frac{k-1}{k}$.
\end{theorem}
\begin{proof}
Denote by $\theta(\ell)$ the number of columns exactly with weight $\ell$ in all construction matrices $\{\mathbf{M}_{i,j}\}_{i\neq j \in [n]}$.
Since $\row(\mathbf{M}_{i,j})=\frac{(n-k)m}{k}$, it is obvious that
$\theta(\ell)=0$ for $\ell > \frac{(n-k)m}{k}$. 
Let $\Theta(\ell)\triangleq\sum_{i=\ell}^{\frac{(n-k)m}{k}}{\theta(i)}$.
We then derive from \eqref{eq:the} that
\begin{equation}\label{eq:the1}
\theta=\frac 1{nm}\sum_{\ell \in [\frac{(n-k)m}{k}]}{\ell \cdot \theta(\ell)}=\frac 1{nm}\sum_{\ell \in [\frac{(n-k)m}{k}]}{\Theta(\ell)} \geq \frac{\Theta(1)+\Theta(2)}{nm}.
\end{equation} 
As Corollaries \ref{co:the2} and \ref{co:the1} imply
$\Theta(2) \geq \frac{(k-1)mn}k$ and $\Theta(1)\geq (n-k)mn$, respectively, 
\eqref{eq:the1} indicates that $\theta \geq n-k+\frac{k-1}{k}$.
\end{proof}

\section{A class of MR-MUB codes with the optimal repair bandwidth}
\label{sec:6}

\subsection{Generic transformation for code construction}


Consider
$(n,k)$ MDS regular array codes with each node having exactly the same number of symbols, denoted as $\alpha$. Hence, $m_i+p_i=\alpha$ for every $i\in[n]$.
Let $\beta_{i}$ be the amount of symbols that needs to be downloaded from all other $n-1$ nodes when repairing node~$i$.  Then, it is known \cite{dimikis2010} that for all $(n,k)$ MDS regular array code designs, $\beta_i\geq \frac{(n-1)\alpha}{(n-k)}$ for every $i\in[n]$. 
As a consequence of this universal lower bound for every $\beta_i$, 
an $(n,k)$ MDS regular array code is said to be with the optimal repair bandwidth 
for all nodes if $\beta_i=\frac{(n-1)\alpha}{(n-k)}$ for every $i\in[n]$.

In 2018, Li et al. \cite{transform} proposed a generic transformation that converts a nonbinary $(n,k)$ MDS regular array code with node size $\alpha$ into another $(n,k)$ MDS regular array code with node size $\alpha'=(n-k)\alpha$ over the same field $\mathbb{F}_{q}$ such that 1) some chosen $(n-k)$ nodes have the optimal repair bandwidth $\frac{(n-1)\alpha'}{(n-k)}=(n-1)\alpha$, and 2) the normalized repair bandwidth\footnote{The normalized repair bandwidth for a node is defined as
\begin{equation}
\frac{\text{the number of symbols downloaded for repairing this node}}{\text{the number of symbols repaired}}.
\end{equation}}
of the remaining $k$ nodes are preserved.
Additionally, after applying the transformation $\lceil \frac{n}{n-k} \rceil$ times,
a nonbinary $(n, k)$ MDS regular array code can be converted 
into an $(n, k)$ MDS regular array code with all nodes achieving the optimal repair bandwidth. 

In this section, using the transformation in \cite{transform}, 
an $(n,k=n-2, 2^{\lceil \frac{n}{n-k} \rceil}m\one)$ regular array code that achieves the optimal repair bandwidth for all nodes is constructed from an $(n,k=n-2,m\one)$ MR-MUB code under $k\mid m$. We will then prove 
in Theorem \ref{th:x9} that
the transformed $(n,k=n-2, 2^{\lceil \frac{n}{n-k} \rceil}m\one)$ regular array code
also have the minimum code redundancy and the minimum update bandwidth and hence is an MR-MUB code.

For completeness, we restate
the generic transform in \cite{transform} in the form that is 
necessary 
in this paper in the following theorem.
Similar to \cite{transform}, the symbols of the codes we construct
are over $\mathbb{F}_{q}$ with $q>2$, where the elements of $\mathbb{F}_{q}$ are denoted as $\{0,1, \gpri,\dots,\gpri^{q-2}\}$ and $\gpri$ is a primitive element of $\mathbb{F}_{q}$. 

\begin{theorem}\label{th:x8}(Generic transform
for $(n=k+2,k)$ regular array codes~\cite{transform})
Let $\mathbf{C}^{(0)}\triangleq[\mathbf{c}_{1}^{(0)} \dots \mathbf{c}_{n}^{(0)}]$ and $\mathbf{C}^{(1)}\triangleq[\mathbf{c}_{1}^{(1)} \dots \mathbf{c}_{n}^{(1)}]$
be codewords of a nonbinary $(n=k+2,k)$ MDS regular array code $\mathcal{C}$ over $\mathbb{F}_{q}$ with node size $\alpha$,
where the data symbols used to generate $\mathbf{C}^{(0)}$ and $\mathbf{C}^{(1)}$ can be different. 
Denote by $\beta_i$ the repair bandwidth of $\mathcal{C}$ for node $i$.
Then, 
\begin{equation}\label{eq:x90}
\mathbf{C}'=\begin{bmatrix}
\mathbf{c}_{1}^{(0)}& \dots & \mathbf{c}_{k}^{(0)}& \mathbf{c}_{k+1}^{(0)}& \mathbf{c}_{k+2}^{(0)}+\mathbf{c}_{k+2}^{(1)}\\
\mathbf{c}_{1}^{(1)}& \dots & \mathbf{c}_{k}^{(1)}& \mathbf{c}_{k+2}^{(0)}+ \gpri\,\mathbf{c}_{k+2}^{(1)}  & \mathbf{c}_{k+1}^{(1)}
\end{bmatrix} \in \mathbb{F}_{q}^{2\alpha \times (k+2)},
\end{equation}
are codewords of an 
$(n=k+2,k)$ MDS regular array code $\mathcal{C}'$ with node size $\alpha'=2\alpha$,
and its repair bandwidth for node $i$ satisfies
\begin{equation}\label{eq:pnrp}
\beta'_{i}=\begin{cases}
2\beta_{i},&\text{for }i \in [k]\\
\frac{(n-1)\alpha'}{n-k}=(n-1)\alpha,&\text{ for }k<i\leq n=k+2.
\end{cases}
\end{equation}
\end{theorem}

It is worth noting that the last two nodes of the transformed code $\mathcal{C}'$
have achieved the universal lower bound and therefore is with the optimal repair bandwidth. Furthermore, it can be inferred from 
\eqref{eq:pnrp} that
if the code $\mathcal{C}$ before transformation is already with the optimal repair bandwidth for every node, then the repair bandwidths of $\mathcal{C}'$ 
are also optimal for all nodes.

\subsection{MR-MUB code construction with the optimal repair bandwidth}

According to Theorem~\ref{th:x8}, given  an $(n,k=n-2,m\one)$ MR-MUB code $\mathcal{C}$ under $k\mid m$, 
we can construct an $(n,n-2,2m\one)$ regular array code $\mathcal{C}'$
that satisfies 1) the last two nodes are with the optimal repair bandwidth,
and 2) the remaining $k$ nodes preserve the same normalized repair bandwidths as their corresponding nodes of $\mathcal{C}$.
In order to distinguish between the codewords before and after transformation, 
we will use 
\begin{equation}
\left\{\mathbf{y}_{i}=\begin{bmatrix}\mathbf{y}_i^{(0)}\\
\mathbf{y}_i^{(1)}
\end{bmatrix}\right\}_{i \in [n]}\quad\text{and}\quad
\left\{\mathbf{q}_{i}=\begin{bmatrix}
\mathbf{q}_i^{(0)}\\
\mathbf{q}_i^{(1)}
\end{bmatrix}\right\}_{i \in [n]}
\end{equation} to denote the data vectors and the parity vectors of the transformed code $\mathcal{C}'$, respectively.
Data vectors and parity vectors of the base code $\mathcal{C}$
are respectively denoted as $\{\mathbf{x}_{i}^{(\ell)}\}_{i \in [n],\ell\in\{0,1\}}$ and $\{\mathbf{p}_{i}^{(\ell)}\}_{i \in [n],\ell\in\{0,1\}}$.
We then have the following correspondence between
$\{\mathbf{y}_{i}^{(\ell)},\mathbf{q}_{i}^{(\ell)}\}_{i \in [n],\ell\in\{0,1\}}$ and $\{\mathbf{x}_{i}^{(\ell)},\mathbf{p}_{i}^{(\ell)}\}_{i \in [n],\ell\in\{0,1\}}$:
\begin{equation}\label{eq:C''}
\begin{aligned}
\mathbf{C}'=&\begin{bmatrix}
\mathbf{y}_{1}^{(0)} & \dots & \mathbf{y}_{k}^{(0)}& \mathbf{y}_{k+1}^{(0)}& \mathbf{y}_{k+2}^{(0)}\\
\mathbf{q}_{1}^{(0)} & \dots & \mathbf{q}_{k}^{(0)}& \mathbf{q}_{k+1}^{(0)}& \mathbf{q}_{k+2}^{(0)}\\
\mathbf{y}_{1}^{(1)} & \dots & \mathbf{y}_{k}^{(1)}& \mathbf{y}_{k+1}^{(1)}& \mathbf{y}_{k+2}^{(1)}\\
\mathbf{q}_{1}^{(1)} & \dots & \mathbf{q}_{k}^{(1)}& \mathbf{q}_{k+1}^{(1)}& \mathbf{q}_{k+2}^{(1)}\\
\end{bmatrix}=\begin{bmatrix}
\mathbf{x}_{1}^{(0)}& \dots & \mathbf{x}_{k}^{(0)}& \mathbf{x}_{k+1}^{(0)}& \mathbf{x}_{k+2}^{(0)}+\mathbf{x}_{k+2}^{(1)}\\
\mathbf{p}_{1}^{(0)}& \dots & \mathbf{p}_{k}^{(0)}& \mathbf{p}_{k+1}^{(0)}& \mathbf{p}_{k+2}^{(0)}+\mathbf{p}_{k+2}^{(1)}\\
\mathbf{x}_{1}^{(1)}& \dots & \mathbf{x}_{k}^{(1)}& \mathbf{x}_{k+2}^{(0)}+ \gpri\,\mathbf{x}_{k+2}^{(1)}  & \mathbf{x}_{k+1}^{(1)}\\
\mathbf{p}_{1}^{(1)}& \dots & \mathbf{p}_{k}^{(1)}& \mathbf{p}_{k+2}^{(0)}+ \gpri\,\mathbf{p}_{k+2}^{(1)}  & \mathbf{p}_{k+1}^{(1)}\\
\end{bmatrix},
\end{aligned}
\end{equation}
which implies
\begin{equation}\label{eq:pnth}
\begin{cases}
\mathbf{x}_{k+1}^{(1)}=\mathbf{y}_{k+2}^{(1)}\\
\mathbf{x}_{k+1}^{(0)}=\mathbf{y}_{k+1}^{(0)}
\end{cases}\quad\text{and}\quad
\begin{cases}
\mathbf{x}_{k+2}^{(1)}=(\gpri-1)^{-1}\big(\mathbf{y}_{k+1}^{(1)}- \mathbf{y}_{k+2}^{(0)}\big)\\
\mathbf{x}_{k+2}^{(0)}=\mathbf{y}_{k+2}^{(0)}-(\gpri-1)^{-1}\big(\mathbf{y}_{k+1}^{(1)}- \mathbf{y}_{k+2}^{(0)}\big)
\end{cases}
\end{equation} 
We then present the main theorem in this section.

\begin{theorem}\label{th:x9}
$\mathcal{C}'$ (whose codewords are defined in \eqref{eq:C''}) is an $(n,k=n-2,2m\one)$ MR-MUB code over $\mathbb{F}_{q}$.
\end{theorem}
\begin{proof}
Recall from \eqref{eq:sl3} that each node of $\mathcal{C}$
has $\alpha=m+p=m+\frac{(n-k)}{k}m=\frac{nm}k$ symbols.
Thus, from \eqref{eq:C''}, 
each node of $\mathcal{C}'$ contains
$\alpha'=2\alpha=\frac{2nm}k$ symbols. 

Since each node of the transformed code $\mathcal{C'}$ has $p'=2p$ parity symbols, its code redundancy achieves the minimum value
given in \eqref{eq:37}. It remains to show $\mathcal{C}'$ also achieves the minimum update bandwidth.


Using the notations in Section~\ref{sec:ub}, 
where the encoding matrices of $\mathcal{C}$ 
are denoted as $\{\mathbf{A}_{i,j}\}_{i,j \in [n]}$ and $\{\mathbf{B}_{i,j}\}_{i,j \in [n]}$,
we consider the update of node $i$ of $\mathcal{C}'$ for $i\in[k]$.
From \eqref{eq:C''}, 
we need to compute 
\begin{equation}
\Delta\mathbf{y}_{i}^{(\ell)}=\mathbf{y}_{i}^{(\ell)*}-\mathbf{y}_{i}^{(\ell)}\quad\text{for }\ell=1,2,
\end{equation}
where we add a star in the superscript to denote the value of a vector
after this updating. 
Then, we must renew $\mathbf{q}_j^{(\ell)}$ for $j \in [k] \setminus \{i\}$ based on $\Delta\mathbf{y}_{i}^{(\ell)}$ according to $\mathbf{q}_{j}^{(\ell)}+\mathbf{B}_{i,j}\mathbf{A}_{i,j}\Delta\mathbf{y}_{i}^{(\ell)}$. The correspondence
in \eqref{eq:C''} then indicates $\mathbf{q}_j^{(\ell)}=\mathbf{p}_j^{(\ell)}$
and $\Delta\mathbf{y}_{i}^{(\ell)}=\Delta\mathbf{x}_{i}^{(\ell)}$, i.e.,
\begin{IEEEeqnarray}{rCl}
\mathbf{q}_{j}^{(\ell)*}&=&\mathbf{p}_{j}^{(\ell)*}=\mathbf{p}_{j}^{(\ell)}+\mathbf{B}_{i,j}\mathbf{A}_{i,j}\Delta\mathbf{x}_{i}^{(\ell)}=\mathbf{q}_{j}^{(\ell)}+\mathbf{B}_{i,j}\mathbf{A}_{i,j}\Delta\mathbf{x}_{i}^{(\ell)}.\label{eq:150pn2}
\end{IEEEeqnarray}
Accordingly, node $i$ shall send both $\mathbf{A}_{i,j}\Delta\mathbf{y}_{i}^{(0)}=\mathbf{A}_{i,j}\Delta\mathbf{x}_{i}^{(0)}
$ and $\mathbf{A}_{i,j}\Delta\mathbf{y}_{i}^{(1)}=\mathbf{A}_{i,j}\Delta\mathbf{x}_{i}^{(1)}$ to node $j \in [k]\setminus \{i\}$, which implies $\gamma'_{i,j}=2\gamma_{i,j}$ for  $i\neq j\in[k]$. 
The renew of $\mathbf{q}_{k+1}$ 
requires sending
$\mathbf{A}_{i,k+1}\Delta\mathbf{x}_{i}^{(0)}$ and $\mathbf{A}_{i,k+2}\big(\Delta\mathbf{x}_{i}^{(0)}+\gpri\,\Delta\mathbf{x}_{i}^{(1)}\big)$
to node $k+1$ since
\begin{equation}
\mathbf{q}_{k+1}^{(0)*}=\mathbf{q}_{k+1}^{(0)}+\mathbf{B}_{i,k+1}\mathbf{A}_{i,k+1}\Delta\mathbf{y}_{i}^{(0)}=\mathbf{p}_{k+1}^{(0)}+\mathbf{B}_{i,k+1}\mathbf{A}_{i,k+1}\Delta\mathbf{x}_{i}^{(0)},
\end{equation}
and
\begin{equation}
\begin{aligned}
\mathbf{q}_{k+1}^{(1)*}=&\mathbf{p}_{k+2}^{(0)*}+ \gpri\,\mathbf{p}_{k+2}^{(1)*}\\
=&\big(\mathbf{p}_{k+2}^{(0)}+\mathbf{B}_{i,k+2}\mathbf{A}_{i,k+2}\Delta\mathbf{x}_{i}^{(0)}\big)+ \gpri\,\big(\mathbf{p}_{k+2}^{(1)}+\mathbf{B}_{i,k+2}\mathbf{A}_{i,k+2}\Delta\mathbf{x}_{i}^{(1)}\big)\\
=&\mathbf{p}_{k+2}^{(0)}+ \gpri\, \mathbf{p}_{k+2}^{(1)}+\mathbf{B}_{i.k+2}\mathbf{A}_{i,k+2}\big(\Delta\mathbf{x}_{i}^{(0)}+\gpri\, \Delta\mathbf{x}_{i}^{(1)}\big)\\
=&\mathbf{q}_{k+1}^{(1)}+\mathbf{B}_{i.k+2}\mathbf{A}_{i,k+2}\big(\Delta\mathbf{x}_{i}^{(0)}+\gpri\,\Delta\mathbf{x}_{i}^{(1)}\big).
\end{aligned}
\end{equation}
Thus, $\gamma'_{i,k+1}=\gamma_{i,k+1}+\gamma_{i,k+2}$ for $i \in [k]$. 
We can similarly obtain $\gamma'_{i,k+2}=\gamma_{i,k+1}+\gamma_{i,k+2}$ for $i \in [k]$ when concerning the adjustment of $\mathbf{q}_{k+2}$
due to the update of $\mathbf{y}_i$.

We next consider the update of $\mathbf{y}_{k+1}$.
Again, we compute $\Delta\mathbf{y}_{k+1}^{(\ell)}=\mathbf{y}_{k+1}^{(\ell)*}-\mathbf{y}_{k+1}^{(\ell)}$ for $\ell=0,1$. 
Note that all of $\mathbf{x}_{k+1}^{(0)}$, $\mathbf{x}_{k+2}^{(0)}$ and 
$\mathbf{x}_{k+2}^{(1)}$ are involved in this update.
Since $\mathbf{y}_{k+2}^{(0)}=\mathbf{x}_{k+2}^{(0)}+\mathbf{x}_{k+2}^{(1)}$
remains unchanged, we have $\Delta\mathbf{y}_{k+2}^{(0)}=\Delta\mathbf{x}_{k+2}^{(0)}+\Delta\mathbf{x}_{k+2}^{(1)}=\mathbf{0}$, 
which together with \eqref{eq:pnth} implies
$\Delta\mathbf{x}_{k+1}^{(0)}=\Delta\mathbf{y}_{k+1}^{(0)}$
and $\Delta\mathbf{x}_{k+2}^{(1)}=(\gpri-1)^{-1}\Delta\mathbf{y}_{k+1}^{(1)}$.
As a result, for $j\in[k]$, the new parity vectors $\mathbf{q}_{j}^\ast$ are renewed according 
to 
\begin{equation}\label{eq:x99}
\begin{aligned}
\mathbf{q}_{j}^{(0)*}=&\mathbf{p}_{j}^{(0)*}\\
=&\mathbf{p}_{j}^{(0)}+\mathbf{B}_{k+1,j}\mathbf{A}_{k+1,j}\Delta\mathbf{x}_{k+1}^{(0)}+\mathbf{B}_{k+2,j}\mathbf{A}_{k+2,j}\Delta\mathbf{x}_{k+2}^{(0)}\\
=&\mathbf{q}_{j}^{(0)}+\mathbf{B}_{k+1,j}\mathbf{A}_{k+1,j}\Delta\mathbf{x}_{k+1}^{(0)}+\mathbf{B}_{k+2,j}\mathbf{A}_{k+2,j}\Delta\mathbf{x}_{k+2}^{(0)}\\
=&\mathbf{q}_{j}^{(0)}+\mathbf{B}_{k+1,j}\mathbf{A}_{k+1,j}\Delta\mathbf{y}_{k+1}^{(0)}-(\gpri-1)^{-1}\mathbf{B}_{k+2,j}\mathbf{A}_{k+2,j}\Delta\mathbf{y}_{k+1}^{(1)}
\end{aligned}
\end{equation}
and
\begin{equation}\label{eq:x100}
\begin{aligned}
\mathbf{q}_{j}^{(1)*}=&\mathbf{p}_{j}^{(1)*}\\
=&\mathbf{p}_{j}^{(1)}+\mathbf{B}_{k+2,j}\mathbf{A}_{k+2,j}\Delta\mathbf{x}_{k+2}^{(1)}\\
=&\mathbf{q}_{j}^{(1)}+\mathbf{B}_{k+2,j}\mathbf{A}_{k+2,j}\Delta\mathbf{x}_{k+2}^{(1)}\\
=&\mathbf{q}_{j}^{(1)}+(\gpri-1)^{-1} \mathbf{B}_{k+2,j}\mathbf{A}_{k+2,j}\Delta\mathbf{y}_{k+1}^{(1)},
\end{aligned}
\end{equation}
which indicates node $k+1$ should send
$\mathbf{A}_{k+1,j}\Delta\mathbf{y}_{k+1}^{(0)}$ and $\mathbf{A}_{k+2,j}\Delta\mathbf{y}_{k+1}^{(1)}$ to node $j$ to renew its parity vector;
hence, $\gamma_{k+1,j}'=\gamma_{k+1,j}+\gamma_{k+2,j}$ for $j \in [k]$. 
Concerning the renew of $\mathbf{q}_{k+2}$,
we derive
\begin{equation}\label{eq:x104}
\begin{aligned}
\mathbf{q}_{k+2}^{(0)*}=&\mathbf{p}_{k+2}^{(0)*}+\mathbf{p}_{k+2}^{(1)*}\\
=&\mathbf{p}_{k+2}^{(0)}+\mathbf{p}_{k+2}^{(1)}+\mathbf{B}_{k+1,k+2}\mathbf{A}_{k+1,k+2}\Delta\mathbf{x}_{k+1}^{(0)}\\
&+\mathbf{B}_{k+2,k+2}\mathbf{A}_{k+2,k+2}\Delta\mathbf{x}_{k+2}^{(0)}+\mathbf{B}_{k+2,k+2}\mathbf{A}_{k+2,k+2}\Delta\mathbf{x}_{k+2}^{(1)}\\
=&\mathbf{q}_{k+2}^{(0)}+\mathbf{B}_{k+1,k+2}\mathbf{A}_{k+1,k+2}\Delta\mathbf{x}_{k+1}^{(0)}+\mathbf{B}_{k+2,k+2}\mathbf{A}_{k+2,k+2}\Delta\mathbf{x}_{k+2}^{(0)}+\mathbf{B}_{k+2,k+2}\mathbf{A}_{k+2,k+2}\Delta\mathbf{x}_{k+2}^{(1)}\\
=&\mathbf{q}_{k+2}^{(0)}+\mathbf{B}_{k+1,k+2}\mathbf{A}_{k+1,k+2}\Delta\mathbf{y}_{k+1}^{(0)}
\end{aligned}
\end{equation}
and
\begin{equation}\label{eq:x105}
\begin{aligned}
\mathbf{q}_{k+2}^{(1)*}=&\mathbf{p}_{k+1}^{(1)*}\\
=&\mathbf{p}_{k+1}^{(1)}+\mathbf{B}_{k+2,k+1}\mathbf{A}_{k+2,k+1}\Delta\mathbf{x}_{k+2}^{(1)}\\
=&\mathbf{q}_{k+2}^{(1)}+\mathbf{B}_{k+2,k+1}\mathbf{A}_{k+2,k+1}\Delta\mathbf{x}_{k+2}^{(1)}\\
=&\mathbf{q}_{k+2}^{(1)}+(\gpri-1)^{-1} \mathbf{B}_{k+2,k+1}\mathbf{A}_{k+2,k+1}\Delta\mathbf{y}_{k+1}^{(1)},
\end{aligned}
\end{equation}
which indicates node $k+1$ should send 
$\mathbf{A}_{k+1,k+2}\Delta\mathbf{y}_{k+1}^{(0)}$ and $\mathbf{A}_{k+2,k+1}\Delta\mathbf{y}_{k+1}^{(1)}$ to node $k+2$;
hence,  $\gamma_{k+1,k+2}'=\gamma_{k+1,k+2}+\gamma_{k+2,k+1}$. 

Last, we consider the update of $\mathbf{y}_{k+2}$, and 
can similarly obtain
$\gamma_{k+2,j}'=\gamma_{k+1,j}+\gamma_{k+2,j}$ for $j\in [k]$ and $\gamma_{k+2,k+1}'=\gamma_{k+1,k+2}+\gamma_{k+2,k+1}$. 

We summarize the matrix of $\gamma_{i,j}'$ for $i\neq j\in[n]$ as follows.
\begin{equation}\label{eq:gamma'}
\begin{aligned}
&\begin{bmatrix}
\gamma_{1,2}' & \dots & \gamma_{1,k}' & \gamma_{1,k+1}' & \gamma_{1,k+2}'\\
\gamma_{2,1}' & \dots & \gamma_{2,k}' & \gamma_{2,k+1}' & \gamma_{2,k+2}'\\
\vdots & \ddots &\vdots &\vdots &\vdots \\
\gamma_{k,1}' & \dots & \gamma_{k,k-1}' & \gamma_{k,k+1}' & \gamma_{k,k+2}'\\
\gamma_{k+1,1}' & \dots & \dots  & \gamma_{k+1,k}' & \gamma_{k+1,k+2}'\\
\gamma_{k+2,1}' & \dots & \dots  & \gamma_{k+2,k}' & \gamma_{k+2,k+1}'\\
\end{bmatrix}\\
&=\begin{bmatrix}
2\gamma_{1,2} & \dots & 2\gamma_{1,k} & \gamma_{1,k+1}+\gamma_{1,k+2} & \gamma_{1,k+1}+\gamma_{1,k+2}\\
2\gamma_{2,1} & \dots & 2\gamma_{2,k} & \gamma_{2,k+1}+\gamma_{2,k+2} & \gamma_{2,k+1}+\gamma_{2,k+2}\\
\vdots & \ddots &\vdots &\vdots &\vdots \\
2\gamma_{k,1} & \dots & 2\gamma_{k,k-1} & \gamma_{k,k+1}+\gamma_{k,k+2} & \gamma_{k,k+1}+\gamma_{k,k+2}\\
\gamma_{k+1,1}+\gamma_{k+2,1} & \dots & \dots  & \gamma_{k+1,k}+\gamma_{k+2,k} & \gamma_{k+1,k+2}+\gamma_{k+2,k+1}\\
\gamma_{k+1,1}+\gamma_{k+2,1} & \dots & \dots  & \gamma_{k+1,k}+\gamma_{k+2,k} & \gamma_{k+1,k+2}+\gamma_{k+2,k+1}\\
\end{bmatrix}.
\end{aligned}
\end{equation}
Since $\mathcal{C}$ is a $(n=k+2,k,m\one)$ MR-MUB code, 
we know from \eqref{eq:sl2} that
$\gamma_{i,j}=\frac{m}{k}$ for $i\neq j\in [n]$. 
We then conclude from \eqref{eq:gamma'} that
$\gamma_{i,j}'=\frac{2m}{k}$ for $i\neq j\in [n]$.
Consequently, $\mathcal{C}'$ is an $(n,n-2,2m\one)$ MR-MUB code over $\mathbb{F}_{q}$, which can be confirmed by Theorem~\ref{the:2}. 
\end{proof}

By Theorem~\ref{th:x8}, we can optimize the repair bandwidth
of two selected nodes at a time, and reapply the transformation $\lceil \frac n{2}\rceil$ times to obtain an $(n,k=n-2,2^{\lceil n/2\rceil}m\one)$ MR-MUB code with optimal repair bandwidth for all nodes as long as $k\mid m$. 

Although the transformation in \cite{transform} holds for general $k$, a further generation of Theorem \ref{th:x9} to general $k$ satisfying, e.g., $k<n-2$, 
cannot be done by following a similar procedure to the current proof, and
the transformed code may not be an MR-MUB code. Hence, what we have proven in Theorem \ref{th:x9} is a particular case that guarantees the transformed code
is an MR-MUB code if the code before transformation is an MR-MUB code.

\section{Conclusion}\label{sec:conclusion}

In this paper, we introduced a new metric, called the {\it update bandwidth}, 
which measures the transmission efficiency in the update process of $(n,k,\mathbf{m})$ irregular array codes in DSSs.
It is an essential measure in scenarios where updates are frequent.
The closed-form expression of the minimum update bandwidth $\gamma_{\min}$
was established (cf.~Theorem \ref{the:2}), and the code parameters,
using which the minimum update bandwidth (MUB) can be achieved, were 
identified. These code parameters then constitute the class of MUB codes. 
As code redundancy is also an important consideration in DSSs,
we next investigated the smallest code redundancy attainable by MUB codes
(cf.~Theorems \ref{th:x4} and \ref{th:xx5}).

We then seek to construct a class of irregular array codes that achieves both the minimum code redundancy and the minimum update bandwidth, named MR-MUB codes. The code parameters for MR-MUB codes are therefore determined (cf.~Theorem \ref{th:xx6}). An interesting result is that under $1<k<n-1$ and $k\mid m_i$ for $i\in[n]$, MR-MUB codes can only be vertical MDS codes unless
$\mathbf{m}=[m_1\ \cdots\ m_n]$ containing only a single non-zero component.
The explicit construction of MR-MUB codes 
was thus focused on $(n,k)$ vertical MDS codes, i.e., $(n,k,m\one)$ MR-MUB codes (cf.~Section \ref{sec:5A}).
A further generalization of the MR-MUB code construction 
was subsequently proposed for a class of MUB codes with the smallest code redundancy (cf.~Section \ref{sec:6b}).

At last, we studied the update complexity and repair bandwidth of MR-MUB codes.
Through the establishment of a lower bound for the update complexity of MR-MUB codes (cf.~Theorem \ref{th:x5}), we found MR-MUB codes may not simultaneously achieve the minimum update complexity. However, 
an $(n,k=n-2,m\one)$ MR-MUB code with the optimal repair bandwidth for all nodes
can be constructed via the transformation in \cite{transform} (cf.~Theorem~\ref{th:x9}).

There are some challenging issues remain unsolved.
\begin{enumerate}
\item Determine the smallest update bandwidth attainable
by irregular MDS array codes~\cite{tosato2014}, defined as the irregular array codes with the minimum code redundancy.

\item Determine the smallest code redundancy attainable 
by MUB codes when the condition of $k\mid m_i$ for $i\in[n]$ is violated.

\item Examine whether $k\mid m$ is also a necessary condition for vertical MDS array codes being MUB codes, provided $k\nmid n$.

\item Check whether the lower bound for the update complexity 
of $(n,k,m\one)$ MR-MUB codes in Theorem~\ref{th:x5} can be improved or achieved.

\item Study the optimal repair bandwidth of MR-MUB codes
under $n-k \geq 3$.
\end{enumerate}

{\bibliographystyle{IEEEtran}
\bibliography{IEEEabrv,update}}

\end{document}